\documentclass[letterpaper, 10 pt, conference]{ieeeconf}

\IEEEoverridecommandlockouts                              

\usepackage{times} 
\usepackage{amsmath} 
\usepackage{amssymb}  
\usepackage{graphicx}

\usepackage{arydshln}

\title{\LARGE \bf
Top-Down Synthesis of Multi-Agent Formation Control:\\
An Eigenstructure Assignment based Approach
}

\author{Takatoshi Motoyama and Kai Cai
\thanks{The authors are with Dept. of Electrical and Information Engineering, Osaka City University, Japan. %
Emails: motoyama@c.info.eng.osaka-cu.ac.jp, kai.cai@eng.osaka-cu.ac.jp. 
}
}

\newtheorem{thm}{Theorem}
\newtheorem{lem}{Lemma}

\newtheorem{prop}{Proposition}
\newtheorem{exmp}{Example}

\newtheorem{prob}{Problem}

\begin{document}

\maketitle
\thispagestyle{empty}
\pagestyle{empty}

\begin{abstract}
We propose a top-down approach for formation control of heterogeneous multi-agent systems, based on the method of {\it eigenstructure assignment}.  Given the problem of achieving scalable formations on the plane, our approach globally computes a state feedback control that assigns desired closed-loop eigenvalues/eigenvectors.  We characterize the relation between the eigenvalues/eigenvectors and the resulting inter-agent communication topology, and design special (sparse) topologies such that the synthesized control may be implemented locally by the individual agents. 
Moreover, we present a hierarchical synthesis procedure that significantly improves computational efficiency.  Finally, we extend the proposed approach to achieve rigid formation and circular motion, and illustrate these results by simulation examples.
\end{abstract}  
\section{Introduction}


Cooperative control of multi-agent systems has been an active research area in the systems control community \cite{JadLinMor:03,SaFaMu:07,RenBea:08,BulCorMar:09,MesEge:10}. Among many problems, {\it formation control} has received much attention \cite{AndYuFidHen:08} owing to its wide applications such as satellite formation flying, search and rescue, terrain exploration, and foraging. 
A main problem studied is stabilization to a {\it rigid} formation, where the goal is to steer the agents to achieve a formation with a specified size and only freedoms of translation and rotation. Several control strategies have been proposed: affine feedback laws \cite{FaxMur:04,Cor:09}, nonlinear gradient-based control \cite{KriBroFra:09,CaoMorYuAnd:11}, and angle-based algorithms \cite{BasBisJen:10}.  Achieving a {\it scalable} formation with unspecified size has also been studied \cite{CooArc:12,LinWanHanFu:14}; a scalable formation may allow the group to adapt to unknown environment with obstacles or targets. In addition, \cite{BaiArcWen:08,DinYanLin:10} have presented methods of controlling formations in motion. 

These different methods for formation control have a common feature in design: namely {\it bottom-up}.  Specifically, the inter-agent communication topology is given {\it a priori}, which defines the neighbors for each agent.  Then based only on the neighborhood information, local control strategies are designed for the individual agents.  The properties of the designed local strategies are finally analyzed at the systemic (i.e. global) level, and correctness is proved under certain graphical conditions on the communication topology. This bottom-up design is indeed the mainstream approach for cooperative control of multi-agent systems that places emphasis on {\it distributed control}.

In this paper, we propose a distinct, {\it top-down} approach for formation control, based on a known method called {\it eigenstructure assignment} \cite{Moore:76,Moore:77,AndShaChu:83,LiuPat:98}.  Different from the bottom-up approach, here there need not be any communication topology imposed {\it a priori} (in fact the agents are typically assumed independent, i.e. uncoupled), and no design will be done at the local level. Indeed, given a multi-agent formation control problem characterized by specific eigenvalues and eigenvectors (precisely defined in Section~II), our approach constructs on the global level a feedback matrix (if it exists) that renders the closed-loop system to possess those desired eigenvalues/eigenvectors, thereby achieving desired formations. Moreover, the synthesized feedback matrix (its off-diagonal entries being zero or nonzero) defines the communication topology, and accordingly the computed feedback control may be implemented by individual agents. Thus our approach features ``compute globally, implement locally''.

The inter-agent communication topology is a {\it result} of control synthesis, rather than given {\it a priori}. 
We characterize the relation between the resulting topology and the eigenstructure chosen for the synthesis.
Further, we show that by appropriately choosing desired eigenvalues and the corresponding eigenvectors, special topologies (star, cyclic, line) can be designed, and the computed feedback control may be implemented locally over these (sparse) topologies.

Although our method requires 
centralized computation of control gain matrices, we show that a straightforward extension of the approach to a {\it hierarchical} synthesis procedure significantly reduces computation time. Empirical evidence is provided to show the efficiency of the proposed hierarchical synthesis procedure; in particular, computation of a feedback control for a group of 1000 agents needs merely a fraction of a second, which is likely to suffice for many practical purposes.

The main advantage of our top-down approach is that it is {\it systematic}, in the sense that it treats heterogeneous agent dynamics and different cooperative control specifications (characterizable by desired eigenstructure) by the same synthesis procedure. We show that scalable formation, rigid formation, and cooperative circular motion can all be addressed using the same method. Additionally, we show that this method is amenable to deal with more general cases where some agents are not self-stabilizable and the initial inter-agent connections are arbitrary.


We first proposed this eigenstructure assignment based approach in \cite{CaiMot:15}, where we applied it to solve the consensus problem. 
Then the conference precursor \cite{MotCai:ACC17} of this paper extended the approach to solve scalable and rigid formations, and proposed a hierarchical synthesis procedure to significantly shorten the computation time.  This paper differs from \cite{MotCai:ACC17} in the following aspects. (i) A precise relation between eigenstructure and topology is characterized (Section~III). (ii) A method for imposing topological constraints on eigenstructure assignment is presented (Section~III.B). (iii) More general cases where the initial inter-agent topology is arbitrary and/or there exist non-stabilizable agents are addressed (Section~IV). (iv) The problem of achieving cooperative circular motion is solved. (v) All the proofs are provided.

We note that \cite{WuIwa:15} also proposed an eigenstructure assignment method and applied it to the multi-agent consensus problem. Their approach is bottom-up: first a communication topology is imposed among the agents, then local control strategies are designed based on eigenstructure assignment respecting the topology, and finally the correctness of the proposed strategies is verified at the global level. By contrast, our approach is top-down: no topology is imposed {\it a priori}, and topology is a result of control synthesis. Moreover, we characterize the relation between topology and eigenstructure, and design special topologies by selecting special eigenstructures. In addition, the problems addressed in the paper are distinct, namely scalable/rigid formation and circular motion on the plane, which involve complex eigenvalues and eigenvectors.

The rest of the paper is organized as follows.
In Section~II we review the basics of eigenstructure assignment and formulate the multi-agent formation control problem. In Section~III we solve the problem by eigenstructure assignment, and discuss the relations between eigenvalues/eigenvectors and topologies. In Section~IV we study the more general cases where the initial inter-agent topology is arbitrary and/or there exist non-stabilizable agents.
In Section~V we present a hierarchical synthesis procedure to reduce computation time, and in Section VI extend the method to achieve rigid formation and circular motion. Simulation examples are given in Section~VII and our conclusions stated in Section~VIII.


\section{PRELIMINARIES AND PROBLEM FORMULATION}
\subsection{Preliminaries on Eigenstructure Assignment}
First we review the basics of eigenstructure assignment~\cite{Moore:76}.
Consider a linear time-invariant finite-dimensional system modeled by
\begin{align}
\label{eq:standard dynamics}
\dot{x}=Ax + Bu
\end{align}
where $x\in\mathbb{C}^n$ is the state vector, $u\in\mathbb{C}^m$ the input vector, and $A \in\mathbb{R}^{n\times n}$, $B \in\mathbb{R}^{n\times m}$.

Suppose we modify (\ref{eq:standard dynamics}) by state feedback $u=Fx$.
It is well-known (e.g. \cite{Won:67}) that
$F$ may be chosen to assign any (self-conjugate) set of closed-loop eigenvalues for $\dot{x}=(A+BF)x$ if and only if $(A,B)$ is controllable. 
Unless $m = 1$ (single input), however, $F$ is not uniquely determined by a set of closed-loop eigenvalues. Indeed, with state feedback $F$ one has additional freedom to assign certain sets of closed-loop eigenvectors. 
Simultaneously assigning both eigenvalues and eigenvectors is referred to as {\it eigenstructure assignment}.

Let $\lambda \in \mathbb{C}$. 
It is shown in~\cite{Moore:76} that if $(A,B)$ is controllable, then there exists

\begin{align}
\label{eq:kernel}
N(\lambda) :=
\begin{bmatrix}
  N_1(\lambda) \\
  N_2(\lambda)
\end{bmatrix}
\in \mathbb{C}^{(n+m) \times m}
\end{align}
with linearly independent columns such that

\begin{align}\label{eq:kernel_eq}
\begin{bmatrix}
\lambda I - A \ B
\end{bmatrix}\begin{bmatrix}
  N_1(\lambda) \\
  N_2(\lambda)
\end{bmatrix} = 0.
\end{align}
Thus the columns of $N(\lambda)$ form a basis of Ker$[\lambda I-A \ \ B]$; Ker denotes {\it kernel}. Also we will use Im to denote {\it image}.

\begin{lem}(\cite{Moore:76}) 
\label{lem:dis_eig}
Consider the system~(\ref{eq:standard dynamics}) and suppose that $(A,B)$ is controllable and Ker$B=0$. 
Let $\{\lambda_1,\ldots,\lambda_n\}$ be a set of \emph{distinct} complex numbers, and $\{v_1,\ldots,v_n\}$ a set of linearly independent vectors in $\mathbb{C}^n$. 
Then there is a unique $F$ such that for every $i\in[1,n]$, $(A+BF)v_i =
\lambda_i v_i$ if and only if
\begin{align} \label{eq:lemma1}
(\forall i \in [1,n]) v_i \in \mbox{Im} N_1(\lambda_i)
\end{align}
where $N_1(\cdot)$ is in (\ref{eq:kernel}).
\end{lem}

Lemma~\ref{lem:dis_eig} provides a necessary and sufficient condition of eigenstructure assignment.
When the condition holds and thus $F$ exists for assigning distinct complex eigenvalues $\lambda_i$ and the corresponding eigenvectors $v_i$ ($i\in[1,n]$), $F$ may be constructed by the following procedure~\cite{Moore:76}.

\medskip

(i) For each $\lambda_i$ compute a basis of
Ker$[\lambda_i I - A \ B]$. Stack the basis vectors to form
$N(\lambda_i)$ in (\ref{eq:kernel}); partition $N(\lambda_i)$ properly
to get $N_1(\lambda_i)$ and $N_2(\lambda_i)$.

(ii) Find $w_i = -N_2(\lambda_i) k_i$, where $k_i \in \mathbb{C}$ is such that
$N_1(\lambda_i) k_i = v_i$ (the condition Ker$B=0$ in
Lemma~\ref{lem:dis_eig} ensures that $N_1(\lambda_i)$ has
independent columns; thus $k_i$ may be uniquely determined).

(iii) Compute $F$ by
\begin{align} \label{eq:dis_real}
F = [w_1 \cdots w_n] [v_1 \cdots v_n]^{-1}.
\end{align}

Note that the entries of $F$ may include complex numbers in general.
If $\{\lambda_1,\ldots,\lambda_n\}$ is a {\it self-conjugate} set of distinct complex numbers, and $v_i=\bar{v}_j$ wherever $\lambda_i=\bar{\lambda}_j$ ($\bar{\cdot}$ denotes complex conjugate), then all entries of $F$ are real numbers.

The procedure (i)-(iii) of computing $F$ has complexity $O(n^3)$, inasmuch as the calculations involved are solving systems of linear equations, matrix inverse and multiplication (e.g. \cite{GolLoa:96}).

We note that the above eigenstructure assignment result may be extended to the case of repeated eigenvalues with generalized eigenvectors.
For details refer to~\cite{Moore:77} or Appendix.

\subsection{Problem Formulation}
Consider a heterogeneous multi-agent system where each agent is modeled by a first-order ODE:

\begin{align}\label{eq:1st-order ODE}
\dot{x}_i = a_ix_i + b_iu_i,~~~i=1,\ldots,n.
\end{align}
Here $x_i \in \mathbb{C}$ is the state variable, $u_i \in \mathbb{C}$ the control variable, $a_i \in \mathbb{R}$ and $b_i (\ne 0) \in \mathbb{R}$ are constant parameters.
Thus each agent is a point mass moving on the complex plane, with possibly stable ($a_i<0$), semistable ($a_i=0$) or unstable ($a_i>0$) dynamics.
The requirement $b_i\ne0$ is to ensure stabilizability$/$controllability of ($a_i, b_i$); thus each agent is stabilizable/controllable.
Note that represented by~(\ref{eq:1st-order ODE}), the agents are independent (i.e. uncoupled) and no inter-agent topology is imposed at this stage.

In vector-matrix form, the system of $n$ independent agents is 
\begin{align} \label{eq:linear system}
\dot{x} = Ax+Bu
\end{align}
where $x:=[x_1 \cdots x_n]^\top \in\mathbb{C}^n$, $u:=[u_1 \cdots u_n]^\top \in\mathbb{C}^n$, $A :=$ diag$(a_1,\ldots,a_n)$ and $B :=$ diag$(b_1,\ldots,b_n)$; here diag$(\cdot)$ denotes a diagonal matrix with the specified diagonal entries.
Consider modifying (\ref{eq:linear system}) by a state feedback $u=Fx$ and thus the closed-loop system is
\begin{align} \label{eq:closed loop system}
\dot{x}=(A+BF)x.
\end{align}
Straightforward calculation shows that the diagonal entries of $A+BF$ are $a_i+b_iF_{ii}$, and the off-diagonal entries $b_iF_{ij}$.
Since $b_i\ne 0$, the off-diagonal entries $(A+BF)_{ij}\ne 0$  if and only if $F_{ij}\ne 0~(i\ne j)$.

In view of the structure of $A+BF$, we can define a corresponding {\it directed graph} $\mathcal{G}=(\mathcal{V},\mathcal {E})$ as follows: 
the \emph{node} set $\mathcal{V}:=\{1, \ldots, n\}$ with node $i \in \mathcal {V}$ standing for agent $i$ (or state $x_i$); 
the \emph{edge} set $\mathcal {E}\subseteq \mathcal {V} \times \mathcal {V}$ with edge $(j,i) \in \mathcal {V}$ if and only if $F$'s off-diagonal entry $F_{ij} \neq 0$ . 
Since $F_{ij} \neq 0$ implies that $x_i$ uses $x_j$ in its state update, we say for this case that agent $j$ {\it communicates} its state $x_j$ to agent $i$, or $j$ is a {\it neighbor} of $i$. 
The graph $\mathcal {G}$ is therefore called a communication network among agents, whose topology is decided by the off-diagonal entries of $F$.
Thus the communication topology is not imposed {\it a priori}, but emerges as the result of applying the state feedback control $u=Fx$.

Now we define the formation control problem of the multi-agent system~(\ref{eq:linear system}).
\begin{prob}
\label{prob:formation problem}
Consider the multi-agent system~(\ref{eq:linear system}) and specify a vector $f\in\mathbb{C}^n$ ($f \neq 0$).
Design a state feedback control $u=Fx$ such that for every initial condition $x(0)$, $\lim_{t\to \infty}x(t)=cf$ for some constant $c\in\mathbb{C}$.
\end{prob}

In Problem~\ref{prob:formation problem}, the specified vector $f$ represents a desired {\it formation configuration} in the plane. 
By formation configuration we mean that the geometric information of the formation remains when scaling and rotational effects are discarded.
Indeed, by writing the constant $c\in\mathbb{C}$ in the polar coordinate form (i.e. $c=\rho e^{j\theta}$, $j=\sqrt{-1}$), the final formation $cf$ is the configuration $f$ scaled by $\rho$ and rotated by $\theta$. The constant $c$ is unknown {\it a priori} and in general depends on the initial condition $x(0)$.
Note also that Problem~1 includes the consensus problem as a special case when $f={\bf 1}:=[1\cdots 1]^\top$.

To solve Problem~\ref{prob:formation problem}, we note the following fact.
\begin{prop}
\label{prop:formation problem}
Consider the multi-agent system~(\ref{eq:linear system}) and state feedback $u=Fx$. If $A+BF$ has a simple eigenvalue $0$, with the corresponding eigenvector $f$, and other eigenvalues have negative real parts, then for every initial condition $x(0)$, $\lim_{t\to \infty}x(t)=cf$ for some $c\in\mathbb{C}$.
\end{prop}

\begin{proof}
The solution of the closed-loop system~(\ref{eq:closed loop system}) is $x(t) = e^{(A+BF)t}x(0)$. 
Since $A+BF$ has a simple eigenvalue $0$, with the corresponding eigenvector $f$, and other eigenvalues have negative real parts, 
it follows from the standard linear systems analysis that $x(t) \rightarrow (w^\top x(0)) f $ as $t \rightarrow \infty$. Here $w \in \mathbb{C}^n$ is the left-eigenvector of $A+BF$ with respect to the eigenvalue $0$. 
Therefore $\lim_{t\to\infty}{x(t)}=cf$, where $c:=w^\top x(0)$.
\end{proof}

In view of Proposition~\ref{prop:formation problem}, if the specified eigenvalues and the corresponding eigenvectors may be assigned by state feedback $u=Fx$, then Problem~1 is solved.
To this end, we resort to eigenstructure assignment.

\section{Main Results}
In this section, we solve Problem~1, the formation control problem of multi-agent systems, by the method of eigenstructure assignment.  The following is our first main result.
\begin{thm} \label{thm:form}
Consider the multi-agent system~(\ref{eq:linear system}) and let $f$ be a desired formation configuration.
Then there always exists a state feedback control $u=Fx$ that solves Problem~1, i.e.
\begin{align*}
(\forall x(0)\in\mathbb{C}^n)(\exists c\in\mathbb{C})
\lim_{t\to\infty} x(t)=cf.\\
\end{align*}
\end{thm}

\begin{proof}
By Proposition~\ref{prob:formation problem}, the multi-agent system (\ref{eq:linear system}) with $u=Fx$ achieves a formation configuration $f\in\mathbb{C}^n$ if the closed-loop matrix $A+BF$ has the following eigenstructure:
(i) its eigenvalues $\{\lambda_1,\lambda_2,\ldots,\lambda_n\}$ satisfy
\begin{align} \label{eq:con_eig}
&\scalebox{0.9}{$\displaystyle
0=\lambda_1<|\lambda_2|\leq \cdots \leq |\lambda_n|,\ \mbox{and}\
(\forall i \in[2,n]) \mbox{Re}(\lambda_i)<0
$}
\end{align}
(ii) the corresponding eigenvectors
\begin{align} \label{eq:con_eigvec}
\{v_1,v_2,\ldots,v_n\} \mbox{ are linearly independent, and $v_1 = f$}
\end{align}
So we must verify that the above eigenstructure is assignable by state feedback $u=Fx$ for the multi-agent system (\ref{eq:linear system}).
Note that except for $(\lambda_1,v_1)$ which is fixed, we have freedom to choose $(\lambda_i,v_i)$, $i\in[2,n]$.
For simplicity we let $\lambda_i$ be all distinct, and thus Lemma~1 can be applied.

In (\ref{eq:linear system}), we have $A =$ diag$(a_1,\ldots,a_n)$, $B =$ diag$(b_1,\ldots,b_n)$, and $b_i \ne 0$ for all $i \in [1,n]$. Thus it is easily checked that the pair $(A,B)$ is controllable and Ker$B=0$.
To show that there exists $F$ such that $(A+BF)v_i = \lambda_i v_i$, for each $i \in [1,n]$ (with $\lambda_i, v_i$ specified above), it suffices to verify the condition~(\ref{eq:lemma1}) of Lemma~1.

First, for $\lambda_1=0$, we find a basis for
\begin{align*}
\mbox{Ker}[\lambda_1 I - A \ \ B] = \mbox{Ker}[ - A \ \ B]
\end{align*}
and derive
$N_1(\lambda_1)=B$ and $N_2(\lambda_1)=A$ ($N_1(\cdot),N_2(\cdot)$ in (\ref{eq:kernel})).
Thus $\mbox{Im} N_1(\lambda_1) = \mathbb{C}^n$, and hence $v_1=f \in \mbox{Im} N_1(\lambda_1)$, i.e. the condition~(\ref{eq:lemma1}) of Lemma~1 holds.

Next, let $i \in [2,n]$; we find a basis for Ker$[\lambda_i I - A \ B]$ and derive $N_1(\lambda_i)=B$ and $N_2(\lambda_i)=A - \lambda_i I $.
So again $\mbox{Im} N_1(\lambda_i) = \mathbb{C}^n$, and $v_i \in \mbox{Im} N_1(\lambda_i)$, i.e. the condition~(\ref{eq:lemma1}) of Lemma~1 holds.
Therefore, we conclude that there always exists a state feedback $u=Fx$ such that the multi-agent system (\ref{eq:linear system}) achieves the formation configuration $f$.
\end{proof}

\smallskip

In the proof we considered distinct eigenvalues $\lambda_i$, $i\in[1,n]$;
hence the control gain matrix $F$ may be computed by (\ref{eq:dis_real}).
The computed $F$ in turn gives rise to the agents' communication graph $\mathcal {G}$.
The following is an illustrative example.

\medskip

\begin{exmp}
Consider the multi-agent system~(\ref{eq:linear system}) of $4$ single integrators (that is, $a_i=0$ and $b_i=1$, $i=1,...,4$).

(i) {\bf Square formation} with $f = [1\ \ j\ \ -1\ \ -j]^\top$ ($j = \sqrt{-1}$). Let the desired closed-loop eigenvalues be $\lambda_1=0$, $\lambda_2=-1$, $\lambda_3=-2$, $\lambda_4=-3$ and the corresponding eigenvectors be $v_1=f$, $v_2=[1\ \ 1\ \  0\ \  0]^\top$, $v_3=[0\ \ 1\ \  1\ \  0]^\top$, $v_4=[0\ \ 0\ \  1\ \  0]^\top$.
By (\ref{eq:dis_real}) one computes the control gain matrix $F_1$, which determines the corresponding communication graph $\mathcal{G}_1$ (see Fig.~\ref{fig:ex1}).

Observe that $F_1$ contains complex entries, which may be viewed as control gains for the real and imaginary axes, respectively,  or scaling and rotating gains on the complex plane. 
Also note that $\mathcal{G}_1$ has a spanning tree with node 4 the root, and the computed feedback control $u=Fx$ can be implemented by the four agents individually.

(ii) {\bf Consensus} with $f = [1 \ 1 \ 1 \ 1]^\top$. Let the desired eigenvalues be $\lambda_1=0$, $\lambda_2=-1$, $\lambda_3=-3$, $\lambda_4=-4$ and the corresponding eigenvectors be $v_1=f$, $v_2=[1\ \ 1\ \ 0\ \  1]^\top$, $v_3=[1\ \ 0\ \ 0\ \  1]^\top$, $v_4=[0\ \ 0\ \ 1\ \ -1]^\top$. Again by (\ref{eq:dis_real}) one computes the control gain matrix $F_2$ and the corresponding graph $\mathcal{G}_2$ (see Fig.~\ref{fig:ex1}). 

Note that in this case $F_2$ is real and $\mathcal{G}_2$ strongly connected. But unlike the usual consensus algorithm (e.g. \cite{SaFaMu:07}), $-F_2$ is not a {\it graph Laplacian matrix} for the entries  $(2,1)$ and $(3,1)$ are positive. Thus our eigenstructure assignment based approach may generate a larger class of consensus algorithms with negative weights.
\end{exmp}
%
%

\medskip

\begin{figure}[t]
  	\centering
  	\includegraphics[width=0.45\textwidth]{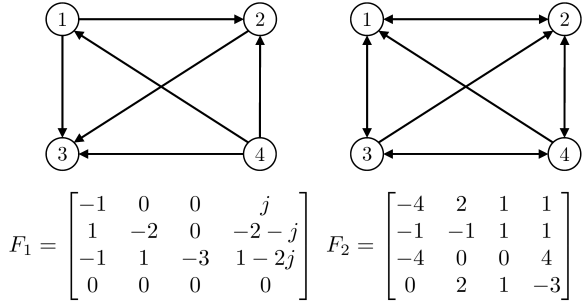}
	\caption{Example~1\label{fig:ex1}}
\end{figure}


We remark that in our approach, the \emph{convergence speed} to the desired formation configuration is assignable.  This is because the convergence speed is dominated by the eigenvalue $\lambda_2$, with the second largest real part, of the closed-loop system $\dot{x}=(A+BF)x$; and in our approach $\lambda_2$ is freely assignable. The smaller the Re$(\lambda_2)$ is, the faster the convergence to formation occurs (at the cost of higher control gain). As an example, for (ii) in Example~1 assign the second largest eigenvalue $\lambda_2=-2$ (originally $-1$), and change $v_4=[0\ 0\ \frac{1}{2}\ -1]$. This results in a new feedback matrix
\begin{align*}
F_2'=
\begin{bmatrix}
   -4& 1& 2& 1\\
   -1&-2& 2& 1\\
   -2& 0& 0& 2\\
    0& 1& 2&-3
\end{bmatrix}
\end{align*}
which has zero entries at the same locations as $F_2$. 
Thus with the same topology, $F'_2$ achieves faster convergence speed.

As we have seen in Example~1, the feedback matrix $F$'s off-diagonal entries, which determine the topology of $\mathcal {G}$, 
are dependent on the choice of eigenvalues as well as eigenvectors.
Namely different sets of eigenvalues and eigenvectors result in different inter-agent communication topologies.
Our next result characterizes a precise relation between the eigenvalues/eigenvectors and the topologies.

\begin{thm} \label{thm:topo}
Consider the multi-agent system~(\ref{eq:linear system}) and $f$ a desired formation configuration.
Let the eigenstructure $\lambda_i$ and $v_i$ ($i=1,...,n$) be as in (\ref{eq:con_eig}) and (\ref{eq:con_eigvec}), and denote the rows of $[v_1 \cdots v_n]^{-1}$ by $v_i^*$. Then the communication graph $\mathcal{G} = (\mathcal{V}, \mathcal{E})$ of the (closed-loop) multi-agent system is such that
\begin{align*}
(i_1,j_1), \ldots, (i_K,j_K) \notin \mathcal{E} \ \ \ (K \geq 1)
\end{align*}
if and only if the vector
\begin{align*}
[\lambda_2 \ \cdots \ \lambda_n]^\top
\end{align*}
is orthogonal to the subspace spanned by the $K$ vectors:
\begin{align*}
[v_{2 \, i_1} v^*_{2 \, j_1} \ \cdots \ v_{n \, i_1} v^*_{n \, j_1}]^\top, \cdots, [v_{2 \, i_K} v^*_{2 \, j_K} \ \cdots \ v_{n \, i_K} v^*_{n \, j_K}]^\top.
\end{align*}
\end{thm}
\medskip

\begin{proof}
For each $\lambda_i$, $i \in [1,n]$, we derive from (\ref{eq:kernel_eq}) that
\begin{align*}
(\lambda_i I - A)N_1(\lambda_i) + BN_2(\lambda_i) = 0
\end{align*}
Choose $N_1(\lambda_i) = B$ and $N_2(\lambda_i)=-(\lambda_i I-A)$ to satisfy the above equation.
Then $k_i=N_1^{-1}(\lambda_i)v_i=B^{-1}v_i$ and $w_i=-N_2(\lambda_i)k_i=(\lambda_i I - A)B^{-1}v_i$. By (\ref{eq:dis_real}) we have
\begin{align*}
F&=[w_1\cdots w_n][v_1\cdots v_n]^{-1} \\
&=[(\lambda_1 I - A)B^{-1}v_1 \cdots (\lambda_n I - A)B^{-1}v_n][v_1\cdots v_n]^{-1} \\
&=B^{-1}(-A[v_1\cdots v_n] + [\lambda_1 v_1 \cdots \lambda_n v_n])[v_1\cdots v_n]^{-1}\\
&=-B^{-1}A + B^{-1}[v_1\cdots v_n]\mbox{diag}(\lambda_1,\ldots,\lambda_n)[v^*_1\cdots v^*_n]^{\top}.
\end{align*}
Thus the closed-loop matrix
\begin{align} \label{eq:A+BF}
A+BF=[v_1\cdots v_n]\mbox{diag}(\lambda_1,\ldots,\lambda_n)[v^*_1\cdots v^*_n]^{\top}.
\end{align}
The $(i,j)$-entry of $A+BF$ is 
\begin{align*}
(A+BF)_{ij} &= \lambda_1 v_{1i} v^*_{1j} + \lambda_2 v_{2i} v^*_{2j} + \cdots + \lambda_n v_{ni} v^*_{nj}  \\
&= \lambda_2 v_{2i} v^*_{2j} + \cdots + \lambda_n v_{ni} v^*_{nj} \ \ \ (\lambda_1=0)\\
&= [\lambda_2 \ \cdots \ \lambda_n] [v_{2i} v^*_{2j} \ \cdots \ v_{ni} v^*_{nj}]^\top.
\end{align*}
Therefore $(A+BF)_{i_1 j_1} = \cdots = (A+BF)_{i_K j_K} = 0$, i.e. in the communication graph $(i_1,j_1), \ldots, (i_K,j_K) \notin \mathcal{E}$, if and only if the vector
$[\lambda_2 \ \cdots \ \lambda_n]^\top$
is orthogonal to each of the following $K$ vectors:
\begin{align*}
[v_{2 \, i_1} v^*_{2 \, j_1} \ \cdots \ v_{n \, i_1} v^*_{n \, j_1}]^\top, \cdots, [v_{2 \, i_K} v^*_{2 \, j_K} \ \cdots \ v_{n \, i_K} v^*_{n \, j_K}]^\top.
\end{align*}
Namely $[\lambda_2 \ \cdots \ \lambda_n]^\top$ is orthogonal to the subspace spanned by these $K$ vectors.
\end{proof}

\smallskip

Once the desired eigenvalues and eigenvectors are chosen, 
Theorem~\ref{thm:topo} provides a necessary and sufficient condition to check the interconnection topology among the agents, without actually computing the feedback matrix $F$. On the other hand, the problem of choosing an appropriate eigenstructure to match a given topology is more difficult, inasmuch as there are many free variables to be determined in the eigenvalues and eigenvectors. While we shall investigate the general problem of eigenstructure design for imposing particular topologies in our future work, in the next subsection, nevertheless,  we show that choosing certain appropriate eigenstructures results in certain special (sparse) topologies. With these topologies the synthesized control $u=Fx$ may be implemented in a distributed fashion.



\subsection{Special Topologies}
We show how to derive the following three types of special topologies by choosing appropriate eigenstructures.
\subsubsection{Star Topology}
A directed graph $\mathcal{G}=(\mathcal{V},\mathcal {E})$ is a star topology if there is a single root node, say node 1, and $\mathcal {E} = \{(1,i) | i \in [2,n]\}$. Thus all the other nodes receive information from, and only from, the root node 1. In terms of the total number of edges, a star topology is one of the sparsest topologies, with the least number ($n-1$) of edges, that contain a spanning tree. Now consider the following eigenstructure.
\begin{align}
&\mbox{eigenvalues: } \lambda_1=0, \lambda_2,\ldots,\lambda_n \mbox{ distinct } \notag\\
& \hspace{1.9cm} \mbox{and }  \mbox{Re}(\lambda_2),\ldots,\mbox{Re}(\lambda_n)<0 \notag\\
&\mbox{eigenvectors: } 
\underset{
\begin{array}{@{}c@{}}
\mbox{(independent)}
\end{array}
}
{[v_1 \ v_2 \cdots v_n]}=
\left[
\begin{array}{cccc}
f_1	&0		&\cdots	&0 \\
f_2	&1		&\cdots	&0 \\
\vdots	&\vdots	&\ddots	&\vdots \\
f_n	&0	&\cdots	&1 \\
\end{array}
\right]\label{eq:star eig}
\end{align}

\begin{prop}\label{prop:star topology}
Consider the multi-agent system (\ref{eq:linear system}).
If the eigenstructure~(\ref{eq:star eig}) is used in the synthesis of feedback control $u=Fx$, then Problem~\ref{prob:formation problem} is solved and the resulting graph $\mathcal{G}$ is a star topology.
\end{prop}
\begin{proof}
First, it follows from the proof of Theorem~\ref{thm:form} that the eigenstructure~(\ref{eq:star eig}) can be assigned to the closed-loop matrix $A+BF$. Then by Proposition~\ref{prop:formation problem}, Problem~\ref{prob:formation problem} is solved. 

Now proceed analogously to the proof of Theorem~\ref{thm:topo} and derive $A+BF$ as in (\ref{eq:A+BF}). 
Substituting into (\ref{eq:A+BF}) the eigenvalues and eigenvectors in (\ref{eq:star eig}), as well as 
\begin{align*}
[v_1 \ v_2 \cdots v_n]^{-1}=
\begin{bmatrix}
\frac{1}{f_1}		&0	&\cdots	&0\\
-\frac{f_2}{f_1}	        &1	&\cdots	&0\\
\vdots			&\vdots	&\ddots	&\vdots \\
-\frac{f_n}{f_1}	        &0		&\cdots	&1\\
\end{bmatrix}
\end{align*}
we derive 
\begin{align*}
A+BF =
\left[
\begin{array}{cccc}
0				&0		&\cdots	&0 \\
-\frac{\lambda_2 f_2}{f_1}	&\lambda_2	& \cdots	&0 \\
\vdots				& \vdots		&\ddots	&\vdots \\
-\frac{\lambda_n f_n}{f_1}	&0		&\cdots	&\lambda_n \\
\end{array}
\right].
\end{align*}
Therefore the corresponding graph $\mathcal{G}$ is a star topology with node 1 the root.
\end{proof}
\smallskip

\subsubsection{Cyclic Topology}
A directed graph $\mathcal{G}=(\mathcal{V},\mathcal {E})$ is a cyclic topology if $\mathcal {E} = \{(1,2), (2,3), ..., (n-1,n),(n,1)\}$. Consider the following eigenstructure.
\begin{align}
\scalebox{0.9}{$\displaystyle $}
&\mbox{eigenvalues: } \{\lambda_1,\lambda_2, \ldots, \lambda_n\} = \{0, \omega -1, \ldots, \omega^{n-1}-1 \} \notag\\
&\mbox{eigenvectors: } 
\underset{
\begin{array}{@{}c@{}}
\mbox{(independent)}
\end{array}
}
{[v_1 \ v_2 \cdots v_n]} = 
\scalebox{0.7}{$\displaystyle
\left[
\begin{array}{cccc}
f_1	&f_1		& \cdots	&f_1 \\
f_2	&f_2\omega	& \cdots	&f_2\omega^{n-1} \\
f_3	&f_3\omega^2		&\cdots	&f_3\omega^{2(n-1)} \\
\vdots	&\vdots		&\vdots	&\vdots \\
f_n	&f_n\omega^{n-1}	&\cdots	&f_n\omega^{(n-1)(n-1)}
\end{array}
\right] \label{eq:cyclic eig}
$}
\end{align}
where $\omega := e^{2\pi j / n}$ ($j=\sqrt{-1}$).
\medskip
\begin{prop}\label{prop:cyclic topology}
Consider the multi-agent system (\ref{eq:linear system}).
If the eigenstructure~(\ref{eq:cyclic eig}) is used in the synthesis of feedback control $u=Fx$, then Problem~1 is solved and the resulting $\mathcal{G}$ is a cyclic topology.
\end{prop}

\begin{proof}
Following the same lines as the proof of Proposition~\ref{prop:star topology}, and using the eigenstructure in (\ref{eq:cyclic eig}) and
\begin{align*}
[v_1 \ v_2 \cdots v_n]^{-1}=\frac{1}{n}
\begin{bmatrix}
\frac{1}{f_1}	&\frac{1}{f_2}				        &\cdots	&\frac{1}{f_n}			\\
\frac{1}{f_1}	&\frac{\bar{\omega}}{f_2}			&\cdots	&\frac{\bar{\omega}^{n-1}}{f_n}		\\
\frac{1}{f_1}	&\frac{\bar{\omega}^2}{f_2}				&\cdots	&\frac{\bar{\omega}^{2(n-1)}}{f_n}	\\
\vdots		&\vdots				                 &					&\vdots				\\
\frac{1}{f_1}	&\frac{\bar{\omega}^{(n-1)}}{f_2}        &\cdots	&\frac{\bar{\omega}^{(n-1)(n-1)}}{f_n}\\
\end{bmatrix}
\end{align*}
we derive
\begin{align*}
A+BF =
\left[
\begin{array}{ccccc}
-1	&\frac{f_1}{f_2}	&0			& \cdots			&0	\\
0			&-1	&\frac{f_2}{f_3}	&\cdots			&0	\\
\vdots			&\vdots			&\ddots			&\ddots	&\vdots 	\\
0			&0			&0			&\ddots	&\frac{f_{n-1}}{f_n}	\\
\frac{f_n}{f_1}		&0			&0			&\cdots	&-1	\\
\end{array}
\right].
\end{align*}
By inspection, we conclude that the corresponding graph $\mathcal{G}$ is a cyclic topology.
\end{proof}
\medskip
\subsubsection{Line Topology}
A directed graph $\mathcal{G}=(\mathcal{V},\mathcal {E})$ is a (directed) line topology if there is a single root node, say node 1, and $\mathcal {E} = \{(1,2), (2,3), ..., (n-1,n)\}$. A line topology is also one of the sparsest topologies containing a spanning tree. Now consider the following eigenstructure.
\begin{align}\label{eq:line eig}
&\mbox{eigenvalues: }  \lambda_1=0, \lambda_2 = \cdots = \lambda_n = -1 \notag\\
&\mbox{eigenvectors: }  
\underset{
\begin{array}{@{}c@{}}
\mbox{(independent)}
\end{array}
}
{[v_1 \ v_2 \cdots v_n]} =
\scalebox{0.9}{$\displaystyle
\left[
\begin{array}{cccc}
f_1	&0			&\cdots	&0 \\
f_2	&0			&\cdots	        &-f_2 \\
\vdots	 &\vdots	&	&\vdots \\
f_{n-1}	&0		&\reflectbox{$\ddots$}	        &-f_{n-1} \\
f_n	&-f_n			&\cdots	&-f_n
\end{array}
\right]
$}
\end{align}

\begin{prop}\label{prop:line topology}
Consider the multi-agent system (\ref{eq:linear system}).
If the eigenstructure~(\ref{eq:line eig}) is used in the synthesis of feedback control $u=Fx$, then Problem~1 is solved and the resulting  $\mathcal{G}$ is a line topology.
\end{prop}

Note that in (\ref{eq:line eig}) we have repeated eigenvalues ($\lambda_2$) and the corresponding 
generalized eigenvectors. As a result, Lemma~\ref{lem:dis_eig} and (\ref{eq:dis_real}) for computing the control gain matrix $F$ cannot be applied to this case. Instead, we resort to the generalized method of \cite{Moore:77}, and provide a proof of Proposition~\ref{prop:line topology} in Appendix.

\subsection{Topology Constrained Eigenstructure Assignment}

We end this section by presenting an alternative approach to imposing topological constraints on eigenstructure assignment. 

Suppose that we have computed by (\ref{eq:dis_real}) a feedback matrix $F$ to achieve a desired formation, i.e. the closed-loop matrix ($A + BF$)'s
eigenvalues $\lambda_1,...,\lambda_n$ and eigenvectors $v_1,...,v_n$ satisfy (\ref{eq:con_eig}) and (\ref{eq:con_eigvec}). 
Now assume that a topological constraint is imposed such that agent $i$ cannot receive information from agent $j$ (for reasons such as cost or physical impossibility). But unfortunately, in the computed $F$ the $(i,j)$-entry $F_{i j} \neq 0$.
Thus our goal is to derive a new feedback matrix $\hat{F}$, by suitably modifying $F$, such that $\hat{F}_{i j}=0$. In doing so, the new closed-loop matrix $A + B\hat{F}$ will generally have different eigenvalues $\hat{\lambda}_1,...,\hat{\lambda}_n$ and eigenvectors $\hat{v}_1,...,\hat{v}_n$.
Hence we must check if these new eigenvalues and eigenvectors still satisfy (\ref{eq:con_eig}) and (\ref{eq:con_eigvec}).

Our approach proceeds as follows, which is inspired by the ``constrained feedback'' method in \cite{AndShaChu:83}.
Writing $V:=[v_1 \ \cdots \ v_n]$ and $J:=\mbox{diag}(\lambda_1,...,\lambda_n)$, we have from (\ref{eq:A+BF}):
\begin{align*}
&(A+BF)V = VJ \\
\Rightarrow & BFV = VJ - AV \\
\Rightarrow & FV = B^{-1} (VJ-AV).
\end{align*}
Let $\psi:=B^{-1}(VJ-AV)$ and denote by $\psi_i$ the $i$th row of $\psi$, also $F_i$ the $i$th row of $F$. 
Then the above equation is rewritten in terms of Kronecker product and row stacking as follows:
\begin{align*}
\begin{bmatrix}
V^\top&&0\\
&\ddots&\\
0&&V^\top
\end{bmatrix}
\begin{bmatrix}
F_1^\top\\
\vdots\\
F_n^\top
\end{bmatrix}=
\begin{bmatrix}
\psi_1^\top\\
\vdots\\
\psi_n^\top
\end{bmatrix}.
\end{align*}
To constrain the $(i,j)$-entry of $F$ to be zero, we focus on 
the equation $V^\top F^\top_i = \psi_i^\top$. 
Deleting $F_{ij}$ from $F^\top_i$ as well as the $j$th column of $V^\top$, 
we obtain the reduced equation:
\begin{align}\label{eq:contraint F equation}
\hat{V}^\top \hat{F}_i^\top= \psi_i^\top
\end{align}
where $\hat{V}^\top\in\mathbb{C}^{n\times (n-1)}$ is the matrix $V^\top$ with $j$th column deleted and $\hat{F}_i^\top\in\mathbb{C}^{n-1}$ the vector $F_i^\top$ with $j$th element deleted.
Now view the entries of $\hat{F}_i^\top$ as the unknowns (i.e. (\ref{eq:contraint F equation}) contains $n$ equations with $n-1$ unknowns $\hat{F}_{i1},...,\hat{F}_{i(j-1)} \ \hat{F}_{i(j+1)},...,\hat{F}_{in}$). Using the pseudoinverse of $\hat{V}^\top$, denoted by $(\hat{V}^\top)^{\dagger}$, we derive  
\begin{align}\label{eq:Fi_hat}
\hat{F}_i^\top = (\hat{V}^\top)^{\dagger} \psi_i^\top.
\end{align}
By (\ref{eq:Fi_hat}) we set the new feedback matrix  
\begin{align}\label{eq:F_hat}
\hat{F} :=
\begin{bmatrix}
&&&F_1&&&\\
&&&\vdots&&&\\
\hat{F}_{i1}&\cdots&\hat{F}_{i(j-1)}&0&\hat{F}_{i(j+1)}&\cdots&\hat{F}_{in}\\
&&&\vdots&&&\\
&&&F_n&&&
\end{bmatrix}
\end{align}
namely $\hat{F}$ is the same as the originally computed $F$ except for the $i$th row replaced by $\hat{F}_i^\top$ computed in (\ref{eq:Fi_hat}) and $\hat{F}_{ij}=0$.\footnote{The above derivation may be readily extended to deal with more than one topological constraint.}

While the new feedback matrix $\hat{F}$ satisfies the imposed topological constraint, 
the eigenstructure of $A+B\hat{F}$ is generally different from that of $A+BF$. Therefore, we must verify if the new eigenvalues/eigenvectors still satisfy (\ref{eq:con_eig}) and (\ref{eq:con_eigvec}), i.e. achieve formation control. This verification need not always be successful, but in case it does turn out successful, we are guaranteed to achieve the desired formation with a feedback matrix satisfying the imposed topological constraint. We illustrate the above method by the following example.

\begin{exmp}
Consider the multi-agent system~(\ref{eq:linear system}) of $5$ single integrators (that is, $a_i=0$ and $b_i=1$, $i=1,...,5$), and the desired formation is simply consensus ($f=[1 \ 1 \ 1 \ 1 \ 1]^\top$). Choose the following eigenvalues and eigenvectors 
\begin{align*}
\scalebox{0.9}{$\displaystyle $}
&\mbox{eigenvalues: } \{\lambda_1,\lambda_2,\lambda_3,\lambda_4,\lambda_5\} = \{0, -1, -2, -3, -4 \}\\
&\mbox{eigenvectors: } [v_1\ v_2 \ v_3\ v_4\ v_5] = 
\scalebox{0.9}{$\displaystyle
\begin{bmatrix}
     1& 0& 0& 0& 0\\
     1& 1& 0& 0& 1\\
     1& 0& 1& 1&-1\\
     1& 0& 0& 1& 0\\
     1& 0& 1
& 0& 1
\end{bmatrix}
$}
\end{align*}
and compute by (\ref{eq:dis_real}) the feedback matrix $F$ to achieve consensus:
\begin{align*}
F=
\begin{bmatrix}
      0&   0&    0&    0&    0\\
    2.5&  -1&  1.5& -1.5& -1.5\\
      2&   0&   -3&    0&    1\\
      3&   0&    0&   -3&    0\\
      3&   0&    1&   -1&   -3
\end{bmatrix}.
\end{align*}

Suppose that the topological constraint is that agent $2$ cannot receive information from agent $4$, i.e. the $(2,4)$-entry of $F$ ($F_{24}=-1.5$) must be set to be zero. For this we first derive equation $V^\top F^\top_i = \psi_i^\top$ as follows:
\begin{align*}
\begin{bmatrix}
     1& 1& 1& 1& 1\\
     0& 1& 0& 0& 0\\
     0& 0& 1& 0& 1\\
     0& 0& 1& 1& 0\\
     0& 1&-1& 0& 1
\end{bmatrix}
\begin{bmatrix}
F_{21}\\
F_{22}\\
F_{23}\\
F_{24}\\
F_{25}
\end{bmatrix}
=
\begin{bmatrix}
0\\
-1\\
0\\
0\\
-4
\end{bmatrix}.
\end{align*}
Deleting $F_{24}$ from $F^\top_i$ and the 4th column of $V^\top$, we obtain the reduced equation (\ref{eq:contraint F equation}):
\begin{align*}
\begin{bmatrix}
     1& 1& 1& 1\\
     0& 1& 0& 0\\
     0& 0& 1& 1\\
     0& 0& 1& 0\\
     0& 1&-1& 1
\end{bmatrix}
\begin{bmatrix}
\hat{F}_{21}\\
\hat{F}_{22}\\
\hat{F}_{23}\\
\hat{F}_{25}
\end{bmatrix}
=
\begin{bmatrix}
0\\
-1\\
0\\
0\\
-4
\end{bmatrix}.
\end{align*}
Solve this equation for the unknowns $\hat{F}_{21}, \hat{F}_{22}, \hat{F}_{23}, \hat{F}_{25}$, 
we compute by (\ref{eq:Fi_hat}):
\begin{align*} \hat{F}_2^\top=
\begin{bmatrix}
\hat{F}_{21}\\
\hat{F}_{22}\\
\hat{F}_{23}\\
\hat{F}_{25}
\end{bmatrix}
=
\begin{bmatrix}
 1.8571\\
-1.4286\\
 0.8571\\
-1.2857
\end{bmatrix}.
\end{align*}
Finally set as (\ref{eq:F_hat}) the new feedback matrix
\begin{align*}
\hat{F}=
\begin{bmatrix}
         0&       0&       0&       0&      0\\
    1.8571& -1.4286&  0.8571&       0&-1.2857\\
         2&       0&      -3&       0&      1\\
         3&       0&       0&      -3&      0\\
         3&       0&       1&      -1&     -3
\end{bmatrix}.
\end{align*}
Observe that (only) the second row of $\hat{F}$ has entries all different from that of the original $F$, and the $(2,4)$-entry $\hat{F}_{24} = 0$. Moreover the eigenstructure of the new closed-loop matrix $A+B\hat{F} (=\hat{F})$ is 
\begin{align*}
\scalebox{0.9}{$\displaystyle $}
&\mbox{eigenvalues: } \{\hat{\lambda}_1,\hat{\lambda}_2,\hat{\lambda}_3,\hat{\lambda}_4,\hat{\lambda}_5\} = \{0, -1.4286, -2, -3, -4 \}\\
&\mbox{eigenvectors: } [\hat{v}_1 \ \hat{v}_2 \ \hat{v}_3 \ \hat{v}_4 \ \hat{v}_5] \\
&\hspace{1.5cm} = 
\scalebox{0.9}{$\displaystyle
\begin{bmatrix}
    1&         0&         0&         0&         0\\
    1&    2.2361&   -1.0477&   -0.8047&   -1.1352\\
    1&         0&   -1.3969&    1.4753&    1.3623\\
    1&         0&         0&    1.4753&         0\\
    1&         0&   -1.3969&         0&   -1.3623
\end{bmatrix}.
$}
\end{align*}
Hence these new eigenvalues/eigenvectors still satisfy (\ref{eq:con_eig}) and (\ref{eq:con_eigvec}), and therefore consensus is achieved despite of the imposed topological constraint. (In fact, since $\hat{\lambda}_2 < \lambda_2$, we have faster convergence with the new $\hat{F}$.)
\end{exmp}

\section{General Multi-Agent Systems}

So far we have considered the multi-agent system in (\ref{eq:linear system}), where the agents are uncoupled and each is (self-) stabilizable (the matrices $A, B$ are diagonal and $B$'s diagonal entries nonzero).  For (\ref{eq:linear system}) we have shown in Theorem~\ref{thm:form} that a state feedback control, based on eigenstructure assignment, always exists to drive the agents to a desired formation. 

More generally, however, the agents may be initially interconnected (owing to physical coupling or existence of communication channels), and/or some agents might not be capable of stabilizing themselves (though they can receive information from others). It is thus of interest to inquire, based on the eigenstructure assignment approach, what conclusions we can draw for formation control in these more general cases.

\subsection{Arbitrary Inter-Agent Connections}

First we consider the case where the agents have arbitrary initial interconnection, while keeping the assumption that they are individually stabilizable. That is, we consider the following multi-agent system
\begin{align} \label{eq:general linear system}
\dot{x} = Ax+Bu
\end{align}
where $x\in\mathbb{C}^n$, $u\in\mathbb{C}^n$, $A \in \mathbb{R}^{n\times n}$ 
and $B=$ diag$(b_1,\ldots,b_n)$ ($b_i \ne 0$).  The matrix $A$ is now an arbitrary real matrix, modeling an arbitrary (initial) communication topology among the agents.
  

It turns out, despite the general $A$ matrix, that the same conclusion as Theorem~\ref{thm:form} holds. 

\begin{thm}
\label{thm:formation problem}
Consider the multi-agent system~(\ref{eq:general linear system}) and let $f$ be a desired formation configuration. 
Then there always exists a state feedback control $u=Fx$ that achieves formation control, i.e.
\begin{align*}
(\forall x(0)\in\mathbb{C}^n)(\exists c\in\mathbb{C})
\lim_{t\to\infty} x(t)=cf.\\
\end{align*}
\end{thm}
\begin{proof}
The proof proceeds similarly to that of Theorem~\ref{thm:form}. First, for the diagonal matrix $B$ in (\ref{eq:general linear system}), we have (regardless of $A$) that the pair $(A,B)$ is controllable and Ker$B=0$.  

It is left to verify if the (distinct) eigenvalues $\lambda_1,...,\lambda_n$ and eigenvectors $v_1,...,v_n$ as in (\ref{eq:con_eig}) and (\ref{eq:con_eigvec}) satisfy the condition~(\ref{eq:lemma1}) of Lemma~1. 
Let $i \in [1,n]$. Since 
\begin{align*}
\begin{bmatrix}
\lambda_i I- A \ B
\end{bmatrix}\begin{bmatrix}
  B \\
  B^{-1} (A-\lambda_i I) B
\end{bmatrix} = 0
\end{align*}
we find a basis for Ker$[\lambda_i I - A \ B]$ and derive 
\begin{align*}
N_1(\lambda_i)&=B \\
N_2(\lambda_i)&=B^{-1} (A-\lambda_i I) B.
\end{align*}
Thus $\mbox{Im} N_1(\lambda_i) = \mathbb{C}^n$, and $v_i \in \mbox{Im} N_1(\lambda_i)$, i.e. the condition~(\ref{eq:lemma1}) of Lemma~1 holds.
Therefore, there always exists a state feedback $u=Fx$ such that 
the multi-agent system (\ref{eq:general linear system}) achieves the formation configuration $f$.
\end{proof}

In the proof of Theorem~\ref{thm:form}, we derived $N_1(\lambda_i)=B$ and $N_2(\lambda_i)=A-\lambda_i I$. Since $A$ was diagonal and diagonal matrices commute, there held
\begin{align*}
\begin{bmatrix}
\lambda_i I- A \ B
\end{bmatrix}\begin{bmatrix}
  B \\
  A-\lambda_i I
\end{bmatrix} = 0.
\end{align*}
For a general $A$ as in Theorem~\ref{thm:formation problem}, 
we have found instead $N_1(\lambda_i)=B$ and $N_2(\lambda_i)=B^{-1} (A-\lambda_i I) B$ that deal with arbitrary $A$ without depending on the commutativity of matrices.

Theorem~\ref{thm:formation problem} asserts that, as long as the agents are individually stabilizable, formation control is achievable by eigenstructure assignment regardless of how the agents are {\it initially} interconnected. The {\it final} topology, on the other hand, is in general determined by the initial connections `plus' additional ones resulted from the chosen eigenvalues/eigenvectors (as has been discussed in Section~III). 
It may also be possible, however, that the initial connections are `decoupled' by the corresponding entries of the synthesized feedback matrix. 
This is illustrated by the following example.

Consider again Example~1(i), but change $A$ from the zero matrix to the following
\begin{align*}
A= 
\begin{bmatrix}
         0&    0.5&         0&         0\\
         0&      0&         0&         0\\
      -0.5&      0&         0&         0\\
         0&      0&         2&         0
\end{bmatrix};
\end{align*}
that is, agents 1 and 2, 3 and 1, 4 and 3 are initially interconnected. 
Assigning the same eigenstructure as in Example~1(i), we obtain the feedback matrix
\begin{align*}
F=
\begin{bmatrix}
  -1&  -0.5&   0  &         j\\
   1&  -2&   0  &  -2 - j\\
  -0.5&   1&  -3&   1 - 2j\\
   0  &   0  &  -2&   0
\end{bmatrix}.
\end{align*}
Then the closed-loop matrix $A+BF$ (where $B$ is the identity matrix) is
\begin{align*}
A+BF=
\begin{bmatrix}
  -1&  0&   0  &         j\\
   1&  -2&   0  &  -2 - j\\
  -1&   1&  -3&   1 - 2j\\
   0  &   0  &  0&   0
\end{bmatrix}
\end{align*}
which is the same as the feedback matrix $F_1$ (as well as the closed-loop matrix) in Example~1(i). 
Thus despite the initial coupling, the final topology turns out to be the same as that of Example~1(i). 
In particular, in the final topology agents 1 and 2, 4 and 3 are {\it uncoupled} -- their initial couplings are `canceled' by the corresponding entries of the feedback matrix $F$.


\subsection{Existence of Non-Stabilizable Agents}
Continuing to consider arbitrary initial topology (i.e. general A), we further assume that some agents cannot stabilize themselves (i.e. the corresponding diagonal entries of $B$ in (\ref{eq:general linear system}) are zero). Equivalently, the non-stabilizable agents have {\it no} control inputs.
In this case, achieving a desired formation is possible only if those non-stabilizable agents may take advantage of information received from others (via connections specified by $A$). This is a problem of global formation stabilization with locally unstabilizable agents, which has rarely been studied in the literature. We aim to provide an answer using our top-down eigenstructure assignment based approach.  

Without loss of generality, assume that only the first $m$ $(<n)$ agents are stabilizable. 
Thus the multi-agent system we consider in this subsection is
\begin{align}\label{eq:simo system}
\dot{x}=Ax+Bu
\end{align}
where $x\in\mathbb{C}^n$, $u\in\mathbb{C}^m$, $A\in\mathbb{R}^{n\times n}$ and 
\begin{align*}
&B=
\begin{bmatrix}
b_1 	 & 	 &\\
         &\ddots  &\\
         &      & b_m \\
         &      &  \\
         &    \textrm{\LARGE 0}        &   
         \\
         & &
\end{bmatrix} \in\mathbb{R}^{n\times m} \ (m<n) \\
&b_i \neq 0, \ i \in [1,m].
\end{align*}

\begin{thm}
\label{thm:nonstab}
Consider the multi-agent system~(\ref{eq:simo system}) and let $f$ be a desired formation configuration.
Also let $\lambda_1,...,\lambda_n$ and $v_1,...,v_n$ be the desired eigenvalues and eigenvectors satisfying (\ref{eq:con_eig}) and (\ref{eq:con_eigvec}). If 

(i) the pair $(A,B)$ is controllable, and

(ii) $\mbox{Im} B \subseteq \mbox{Im} (A-\lambda_i I)$ for all $i \in [1,n]$, and 

(iii) $v_i \in \mbox{Im} N_1(\lambda_i)$ for all $i \in [1,n]$, where $N_1(\lambda_i)$ satisfies $(A-\lambda_i I)N_1(\lambda_i)=B$, 

\noindent then there exists a state feedback control $u=Fx$ that achieves formation control, i.e.
\begin{align*}
(\forall x(0)\in\mathbb{C}^n)(\exists c\in\mathbb{C})
\lim_{t\to\infty} x(t)=cf.\\
\end{align*}
\end{thm}

\begin{proof}
First observe from (\ref{eq:simo system}) that Ker$B=0$, since $B$'s columns are linearly independent. 
Now let $i \in [1,n]$. Since $(A,B)$ is controllable (condition~(i)), there exist $N_1(\lambda_i)$ and $N_2(\lambda_i)$ such that (\ref{eq:kernel_eq}) holds. Setting $N_2(\lambda_i)=I$, we derive from (\ref{eq:kernel_eq}) the following matrix equation
\begin{align} \label{eq:thm4}
\left(
A -\lambda_i I 
\right)
N_1(\lambda_i) = 
B.
\end{align}
Since $\mbox{Im} B \subseteq \mbox{Im} (A-\lambda_i I)$ (condition~(ii)), 
this equation has a solution $N_1(\lambda_i)$ (which is determined by $A, B, \lambda_i$).  
Finally, since $v_i \in \mbox{Im} N_1(\lambda_i)$ (condition~(iii)), the condition~(\ref{eq:lemma1}) of Lemma~\ref{lem:dis_eig} is satisfied. 
Therefore the desired eigenvalues and eigenvectors satisfying (\ref{eq:con_eig}) and (\ref{eq:con_eigvec}) may be assigned by a state feedback control $u=Fx$, i.e. formation control is achieved.
\end{proof}

Theorem~\ref{thm:nonstab} provides sufficient conditions to ensure solvability of the formation control problem for multi-agent systems with non-stabilizable agents. In the following we illustrate this result by working out a concrete example, where $A$ represents a directed line topology and there is only one agent that is stabilizable (i.e. $B$ is simply a vector). 

\begin{exmp}
Consider the multi-agent system $\dot{x} = Ax+Bu$ with 
\begin{align*} 
A=
\begin{bmatrix}
a_1 	 & 0	&	 &0  \\
\hat{a}_2&a_2	&	 &   \\
	 &\ddots&\ddots  &   \\
0	 &	&\hat{a}_n&a_n
\end{bmatrix}, \ \ \  
B = 
\begin{bmatrix}
b_1\\
0\\
\vdots\\
0
\end{bmatrix}
\end{align*}
where $a_1,...,a_n,\hat{a}_2,...,\hat{a}_n,b_1$ are nonzero.
Namely $A$ represents a directed line topology with agent 1 the root, and $B$ means that only agent 1 is stabilizable. 
Thus this is a {\it single-input} multi-agent system -- by controlling only the root of a directed line. 

First, it is verified that $(A,B)$ is controllable, i.e. condition~(i) of Theorem~\ref{thm:nonstab} is satisfied. 
To ensure condition~(ii), $\mbox{Im} B \subseteq \mbox{Im} (A-\lambda_i I)$, it suffices to choose each desired eigenvalue 
$\lambda_i$ ($i \in [1,n]$) such that 
$\lambda_i \neq a_j$ for $j\in [1,n]$ (i.e. every eigenvalue is distinct from the nonzero diagonal entries of $A$).
At the same time, these eigenvalues must satisfy (\ref{eq:con_eig}).

Having condition~(ii) hold, equation~(\ref{eq:thm4}) has a solution $N_1(\lambda_i)$. Let us solve (\ref{eq:thm4})
\begin{align*}
\begin{bmatrix}
a_1 - \lambda_i 	 &0 	&	 &0  \\
\hat{a}_2&a_2- \lambda_i	&	 &   \\
	 &\ddots&\ddots  &   \\
0	 &	&\hat{a}_n&a_n- \lambda_i
\end{bmatrix}
N_1(\lambda_i)
=
\begin{bmatrix}
b_1\\
0\\
\vdots\\
0
\end{bmatrix}
\end{align*}
and obtain an explicit solution
\begin{align*}
N_1(\lambda_i)
=
\begin{bmatrix}
\frac{b_1}{a_1 - \lambda_i}\\
-\frac{\hat{a}_2 b_1}{(a_1- \lambda_i) (a_2- \lambda_i)}\\
\vdots\\
(-1)^{n-1}\frac{\hat{a}_2\cdots\hat{a}_n b_1}{(a_1- \lambda_i)(a_2- \lambda_i)\cdots (a_n- \lambda_i)}
\end{bmatrix}.
\end{align*}
Hence to ensure condition~(iii) of Theorem~\ref{thm:nonstab}, 
we must choose each desired eigenvector 
$v_i$ ($i \in [1,n]$) such that
$v_i \in \mbox{Im} N_1(\lambda_i)$ and (\ref{eq:con_eigvec}) is satisfied. 
In particular, for $i=1$ we have $\lambda_i=0$ and $v_1=f$; thus $v_1 \in \mbox{Im} N_1(\lambda_1)$ means that the formation 
vector $f$ must be such that 
\begin{align} \label{eq:SIMO}
f = c
\begin{bmatrix}
\frac{b_1}{a_1} &
-\frac{\hat{a}_2b_1}{a_1a_2}&
\cdots &
(-1)^{n-1}\frac{\hat{a}_2\cdots\hat{a}_nb_1}{a_1a_2\cdots a_n}
\end{bmatrix}^{\top}
\end{align}
where $c \in \mathbb{C}$ ($c \neq 0$). This characterizes the set of all achievable formation configurations for the single-input multi-agent system under consideration. 

We conclude that, by controlling only one agent, indeed the root agent of a directed line topology, it is not possible to achieve arbitrary formation configurations but those determined by the nonzeros entries of the matrices $A$ and $B$ in the specific manner as given in (\ref{eq:SIMO}).  

\end{exmp}


\section{Hierarchical Eigenstructure Assignment}

In the previous sections, we have shown that a control gain matrix $F$ can always be computed (as long as every agent is stabilizable) such that the multi-agent formation Problem~1 is solved. Computing such $F$ by (\ref{eq:dis_real}) (see the eigenstructure assignment procedure in Section~II) has complexity $O(n^3)$, where $n$ is the number of agents. Consequently the computation cost becomes expensive as the number of agents increases.   

To address this issue of centralized computation, we propose in this section a {\it hierarchical synthesis procedure}. We shall show that the control gain matrix $F$ computed by this hierarchical procedure again solves Problem~1, which moreover significantly improves computational efficiency (empirical evidence provided in Section~VII).  

For clarity of presentation, let us return to consider the multi-agent system~(\ref{eq:linear system}), and Problem~1 with the desired formation configuration $f \in \mathbb{C}^{n}$ ($f \neq 0$).  Partition the agents into $l \, (\geq 1)$ pairwise disjoint groups. Let group $k \, (\in [1,l])$ have $n_k \, (\geq 1)$ agents; $n_k$ may be different and $\Sigma_{k=1}^l n_k = n$. 

Now for the configuration $f$ and $x, u, A, B$ in (\ref{eq:linear system}), write in accordance with the partition (possibly with reordering)
\begin{align*}
&f =
\begin{bmatrix}
g_1 \\
\vdots \\
g_l
\end{bmatrix},
x =
\begin{bmatrix}
y_1 \\
\vdots \\
y_l
\end{bmatrix},  
u=
\begin{bmatrix}
w_1 \\
\vdots \\
w_l
\end{bmatrix}, \\
&A =
\begin{bmatrix}
A_1 & & \\
& \ddots & \\
& & A_l
\end{bmatrix},
B =
\begin{bmatrix}
B_1 & & \\
& \ddots & \\
& & B_l
\end{bmatrix}
\end{align*}
where $g_k, y_k, w_k \in \mathbb{C}^{n_k}$ and $A_k, B_k \in \mathbb{C}^{n_k \times n_k}$, $k \in [1,l]$. Thus for each group $k$, the dynamics is
\begin{align}
\label{eq:group_dynamics}
\dot{y}_k=A_k y_k + B_k w_k.
\end{align}

For later use, also write $g_{k1}, y_{k1}, w_{k1}$ (resp. $A_{k1}, B_{k1}$) for the first component of $g_k$, $y_k$, $w_k$ (resp. (1,1)-entry of $A_k, B_k$), and $g_0 := [g_{11} \cdots g_{l1}]^\top$, $y_0 := [y_{11} \cdots y_{l1}]^\top$, $w_0 := [w_{11} \cdots w_{l1}]^\top$, $A_0 := \mbox{diag}(A_{11},\cdots,A_{l1})$, $B_0 := \mbox{diag}(B_{11},\cdots,B_{l1})$.

The vector $g_k$ ($k \in [1,l]$) is the {\it local} formation configuration for group $k$, while $g_0$ is the formation configuration for the set of the first component agent from each group. We assume that these configurations are all nonzero, i.e. $g_k \neq 0$ for $k \in [1,l]$ and $g_0 \neq 0$. 
Now we present the hierarchical synthesis procedure.

(i) For each group $k \in [1,l]$ and its dynamics (\ref{eq:group_dynamics}), compute $F_k$ by (\ref{eq:dis_real}) such that 
$A_k+B_kF_k$ has a simple eigenvalue $0$ with the corresponding eigenvector $g_k$, and other eigenvalues have negative real parts;
moreover the topology defined by $F_k$ has a unique root node $y_{k1}$ (e.g. star or line by the method given in Section~III.A).

(ii) Treat $\{y_{k1} | k \in[1,l]\}$ (the group leaders) as a higher-level group, with the dynamics
\begin{align}
\label{eq:highlevel_dynamics}
\dot{y}_0=A_0 y_0 + B_0 w_0.
\end{align}
Compute $F_0\in\mathbb{C}^{l \times l}$ by (\ref{eq:dis_real}) such that 
$A_0+B_0F_0$ has a simple eigenvalue $0$ with the corresponding eigenvector $g_0$, and other eigenvalues have negative real parts.

(iii) Set the control gain matrix $F := F^{\mbox{low}} + F^{\mbox{high}}$, where 
\begin{align*}
F^{\mbox{low}}  :=
\begin{bmatrix}
F_1 & & \\
& \ddots & \\
& & F_l
\end{bmatrix}
\end{align*}
and $F^{\mbox{high}}$ is partitioned according to $F^{\mbox{low}}$, with each block $(i,j)$, $i,j\in[1,l]$
\begin{align*}
(F^{\mbox{high}})_{ij} &= (F_0)_{ij} \cdot
\begin{bmatrix}
1 & 0 & \cdots & 0 \\
0 & 0 & \cdots & 0 \\
\vdots &  \vdots  & \ddots & \vdots \\
0 & 0 & \cdots & 0
\end{bmatrix} \\
&=  
\begin{bmatrix}
(F_0)_{ij} & 0 & \cdots & 0 \\
0 & 0 & \cdots & 0 \\
\vdots &  \vdots  & \ddots & \vdots \\
0 & 0 & \cdots & 0
\end{bmatrix}.
\end{align*}

The computational complexity of Step~(i) is $O(\hat{n}^{3})$, where $\hat{n} := \max\{n_1,...,n_l\}$; and Step~(ii) is $O(l^{3})$. Let $\tilde{n} := \max\{ \hat{n}, l \}$. Then the complexity of the entire hierarchical synthesis procedure is $O(\tilde{n}^{3})$. With proper group partition, this hierarchical procedure can significantly reduce computation time, as demonstrated by an empirical study in Section~VII. 

Note that in Step~(i) of the above procedure, requiring the topology defined by each $F_k$ to have a unique root, i.e. a single leader, is for simplicity of presentation.  It can be extended to the case of multiple leaders, and then in Step~(ii) treat all the leaders at the higher-level. On the other hand, the number of leaders should be kept small such that the high-level control synthesis in Step~(ii) can be done efficiently.

The correctness of the hierarchical synthesis procedure is asserted in the following.

\begin{thm}
Consider the multi-agent system~(\ref{eq:linear system}) and let $f$ be a desired formation configuration.
Then the state feedback control $u=Fx$ synthesized by the hierarchical synthesis procedure solves Problem~1, i.e.
\begin{align*}
(\forall x(0)\in\mathbb{C}^n)(\exists c\in\mathbb{C})
\lim_{t\to\infty} x(t)=cf.\\
\end{align*}
\end{thm}

\begin{proof}
For each $k \in [1,l]$ let $y'_k := [y_{k2} \ \cdots \ y_{kn_k}]^\top$ and $g'_k := [g_{k2} \ \cdots \  g_{kn_k}]^\top \in \mathbb{C}^{n_k - 1}$. 
Thus $y'_k$ and $g'_k$ are $y_k, g_k$ with the first element removed.
By Step~(i) of the hierarchical synthesis procedure, since $y_{k1}$ is the unique root node, we can write $\dot{y}_k = (A_k+B_kF_k)y_k$ as follows:
\begin{align*}
\left[
\begin{array}{c}
\dot{y}_{k1} \\ \hline
 \\
\dot{y}'_{k} \\
 \\
\end{array}
\right] =
\left[
\begin{array}{c|ccc}
0 & \hspace{-0.34cm}| & 0 & \\ \hline
   & &    & \\
H_k & & G_k & \\
   & &    &
\end{array}
\right]
\left[
\begin{array}{c}
y_{k1} \\ \hline
 \\
y'_{k} \\
 \\
\end{array}
\right].
\end{align*}
Then by the eigenstructure of $A_k+B_kF_k$, all the eigenvalues of $G_k$ have negative real parts and 
\begin{align} \label{eq:hier_proof}
\left[
\begin{array}{c|ccc}
H_k & & G_k & 
\end{array}
\right]
\left[
\begin{array}{c}
g_{k1} \\ \hline
 \\
g'_{k} \\
 \\
\end{array}
\right] = 0.
\end{align}
Reorder $x = [y_1^\top \ \cdots \ y_l^\top]^\top$ to get $\hat{x} := [y_0^\top \ {y^{\prime}_1}^{\top} \cdots \  {y^{\prime}_n}^{\top}]^\top$. Then there is a permutation matrix that similarly transforms the control gain matrix $F$ in Step~(iii) to $\hat{F}$, and $\dot{\hat{x}} = \hat{F} \hat{x}$ is
\begin{align*}
\left[
\begin{array}{c}
\dot{y}_{0} \\ \hline
\dot{y}'_{1} \\
\vdots \\
\dot{y}'_{l}
\end{array}
\right] =
\left[
\begin{array}{c|ccc}
       A_0+B_0F_0        &         & 0 & \\ \hline
\hspace{-1.1cm} H_1                              & G_1 &    & \\
              \ddots             &         & \ddots & \\
                \hspace{1.1cm}               H_l &        &     & G_l 
\end{array}
\right]
\left[
\begin{array}{c}
y_{0} \\ \hline
y'_{1} \\
\vdots \\
y'_{l}
\end{array}
\right].
\end{align*}
It then follows from the eigenstructure of $A_0+B_0F_0$ assigned in Step~(ii) and (\ref{eq:hier_proof}) above that the matrix $\hat{F}$ has a simple eigenvalue $0$ with the corresponding eigenvector $\hat{f} := [g_0^\top \ g_1^{'\top} \ \cdots \ g_l^{'\top}]^\top$, and other eigenvalues have negative real parts. Hence (cf. Proposition~\ref{prop:formation problem}), 
\begin{align*}
(\forall \hat{x}(0)\in\mathbb{C}^n)(\exists \hat{c}\in\mathbb{C})
\lim_{t\to\infty} \hat{x}(t)=\hat{c} \hat{f}.
\end{align*}
Since $\hat{x}$ (resp. $\hat{f}$) is just a reordering of $x$ (resp. $f$), the conclusion follows and the proof is complete.
\end{proof}

\section{Rigid Formation and Circular Motion}

In this section we show that our method of eigenstructure assignment may be easily extended to address problems of rigid formation and circular motion.

\subsection{Rigid Formation}

First, we extend our method to study the problem of achieving a rigid formation, one that has translational and rotational freedom but fixed size.
\begin{prob}
\label{prob:rigidformation problem}
Consider the multi-agent system~(\ref{eq:linear system}) and specify $f\in\mathbb{C}^n$ ($f \neq 0$) and $d >0$.
Design a control $u$ such that for every initial condition $x(0)$, $\lim_{t\to \infty}x(t)=c{\bf 1} + d f e^{j\theta}$ for some $c \in \mathbb{C}$ and $\theta \in [0, 2\pi)$.
\end{prob}

In Problem~2, the goal of the multi-agent system (\ref{eq:linear system}) is to achieve a {\it rigid} formation $d f$, with translational freedom in $c$, rotational freedom in $\theta$, and fixed size $d$.

We now present the {\it rigid-formation synthesis procedure}.

(i) Compute $F$ by (\ref{eq:dis_real}) such that 
$A+BF$ has {\it two} eigenvalues $0$ with the corresponding (non-generalized) eigenvectors {\bf 1} and $f$, and other eigenvalues have negative real parts;\footnote{For repeated eigenvalues with non-generalized eigenvectors, the eigenstructure assignment result Lemma~\ref{lem:dis_eig} and the computation of control gain matrix $F$ in (\ref{eq:dis_real}) remain the same as for the case of distinct eigenvalues.} moreover the topology defined by $F$ is 2-rooted\footnote{A 2-rooted topology is one where there exist 2 nodes from which every other node $v$ can be reached by a directed path after removing an arbitrary node other than $v$ \cite{LinWanHanFu:14}.} with exactly 2 roots (say nodes 1 and 2). This topology may be achieved by assigning appropriate eigenstructures, e.g. 
\begin{align} 
&\mbox{eigenvalues: } \lambda_1=\lambda_2=0, \lambda_3,\ldots,\lambda_n \mbox{ distinct } \notag\\
& \hspace{1.9cm} \mbox{and }  \mbox{Re}(\lambda_3),\ldots,\mbox{Re}(\lambda_n)<0 \notag\\
&\mbox{eigenvectors: } 
\underset{
\begin{array}{@{}c@{}}
\mbox{(independent)}
\end{array}
}
{[v_1 \ v_2 \ v_3 \cdots v_n]}=
\left[
\begin{array}{ccccc}
1 & f_1 & 0 &\cdots & 0 \\
1 & f_2	&0		&\cdots	&0 \\
1 & f_3	&1		&\cdots	&0 \\
\vdots & \vdots	&\vdots	&\ddots	&\vdots \\
1 & f_n	&0	&\cdots	&1 \\
\end{array}
\right] 
\label{eq:2rooted}
\end{align}

(ii) Let $f_1,f_2$ be the first two components of $f$, and set
{\small \begin{align*}
\begin{bmatrix}
\dot{x}_1 \\
\dot{x}_2
\end{bmatrix} = 
\begin{bmatrix}
(x_2-x_1)(||x_2-x_1||^2-d^2|f_2-f_1|^2) \\
(x_1-x_2)(||x_1-x_2||^2-d^2|f_1-f_2|^2)
\end{bmatrix}
=: r(x_1,x_2).
\end{align*}}

(iii) Set the control 
\begin{align} \label{eq:rigidformation_control}
u := Fx + 
B^{-1}\left[
\begin{array}{c}
r(x_1,x_2) \\ \hline
0
\end{array}
\right].
\end{align}

The idea of the above synthesis procedure is to first use eigenstructure assignment to achieve a desired formation configuration with two leaders, and then control the size of the formation by stabilizing the distance between the two leaders to the prescribed $d$. The latter is inspired by \cite{LinWanHanFu:14}.
Our result is the following.

\begin{prop} \label{prop:rigidformation}
Consider the multi-agent system (\ref{eq:linear system}) and let $f \in \mathbb{C}^n$, $d >0$. Then the control $u$ in (\ref{eq:rigidformation_control}) synthesized by the rigid-formation synthesis procedure solves Problem~2 for all initial conditions $x(0)$ with $x_1(0) \neq x_2(0)$.
\end{prop}
\begin{proof}
First, by a similar argument to that in the proof of Theorem~\ref{thm:form}, we can show that the desired eigenvalues/eigenvectors (two eigenvalues at 0 with eigenvectors ${\bf 1}$ and $f$; all other eigenvalues with negative real parts) may always be assigned for the multi-agent system (\ref{eq:linear system}).  As a result (cf. Proposition~\ref{prop:formation problem}), 
\begin{align*}
(\forall x(0)\in\mathbb{C}^n)(\exists c,c'\in\mathbb{C})
\lim_{t\to\infty} x(t)=c{\bf 1} + c' f.
\end{align*}
Moreover, choosing the eigenstructure in (\ref{eq:2rooted})
and following similarly to the proof of Proposition~\ref{prop:star topology}, we can show that the resulting topology defined by $F$ is 2-rooted with nodes 1 and 2 the only two roots.

With the 2-rooted topology and the design in Step~(ii), it follows from \cite[Theorem~4.4]{LinWanHanFu:14} that for all $x(0)$ with $x_1(0) \neq x_2(0)$, we have $c' = d f e^{j\theta}$ for some $\theta \in [0, 2\pi)$. 
\end{proof}

An illustrative example of achieving rigid formations is provided in Section~VII below.

\subsection{Circular Motion}
We apply the eigenstructure assignment approach to solve a cooperative circular motion problem, in which the agents all circle around the same center while keeping a desired formation configuration. This cooperative task may find useful applications in target tracking and encircling (e.g. \cite{KimSug:07,LanYanLin:10}). 

\begin{prob}
\label{prob:circularmotion problem}
Consider the multi-agent system~(\ref{eq:linear system}) and specify $f\in\mathbb{C}^n$ ($f \neq 0$) and $b \in \mathbb{R}$ ($b \neq 0$).
Design a state feedback control $u=Fx$ such that for every initial condition $x(0)$, $\lim_{t\to \infty}x(t)=c{\bf 1} + c' f e^{bjt}$ for some $c,c' \in \mathbb{C}$ and $j=\sqrt{-1}$.
\end{prob}

In Problem~3, the goal is that the agents of (\ref{eq:linear system}) all circle around the same center $c$ at rate $b$, while keeping the formation configuration $f$ scaled by $|c'|$. 

Our result is the following.

\begin{prop} \label{prop:circularmotion}
Consider the multi-agent system~(\ref{eq:linear system}) and let $f \in \mathbb{C}^n$, $b \in \mathbb{R}$.
Then there always exists a state feedback control $u=Fx$ that solves Problem~3.
\end{prop}

\begin{proof}
By a similar argument to that in the proof of Theorem~\ref{thm:form}, we can show that for (\ref{eq:linear system}) there always exists $F$ such that $(A+BF)$ has the following eigenstructure: 
\begin{align*}
&\mbox{eigenvalues: } \lambda_1=0, \lambda_2=bj, \lambda_3,\ldots,\lambda_n \mbox{ distinct } \notag\\
& \hspace{1.9cm} \mbox{and }  \mbox{Re}(\lambda_3),\ldots,\mbox{Re}(\lambda_n)<0 \notag\\
&\mbox{eigenvectors: } v_1 = {\bf 1}, v_2=f, \{v_1,v_2,...,v_n\} \mbox{ independent}
\end{align*}
Then (cf. Proposition~\ref{prop:formation problem}),
\begin{align*}
(\forall x(0)\in\mathbb{C}^n)(\exists c,c'\in\mathbb{C})
\lim_{t\to\infty} x(t)=c{\bf 1} + c' f e^{bjt}.
\end{align*}
That is, Problem~3 is solved.
\end{proof}

The key point to achieving circular motion is to assign one, and only one, pure imaginary eigenvalue $bj$, associated with the formation vector $f$. The circular motion is counterclockwise if $b>0$, and clockwise if $b<0$.
One may easily speed up or slow down the circular motion by specifying the value $|b|$. 

Also note that, by a similar synthesis procedure to that for rigid formation in the previous subsection, the multi-agent system~(\ref{eq:linear system}) can be made to achieve circular motion while keeping a rigid formation with some specified size $d >0$.

Circular motion may be applied to the task of target encircling, which is illustrated by an example in the next section.


\section{SIMULATIONS}

We illustrate the eigenstructure assignment based approach by several simulation examples. 
For all the examples, we consider the multi-agent system (\ref{eq:linear system}) with 5 heterogeneous agents, where 
\begin{align*}
A&=\mbox{diag}(1.6,4.7, 3.0,-0.7,-4.2) \\
B&=\mbox{diag}(0.2,1.5,-0.5,-3.3,-3.7). 
\end{align*}
Thus the first 3 agents are unstable while the latter 2 are stable; all agents are stabilizable.

First, to achieve a scalable (regular) pentagon formation, assign the following eigenstructure:
\begin{align*}
&\mbox{eigenvalues:~} \{\lambda_1,\ldots,\lambda_5\}=\{0,-1,-2,-3,-4\}\\
&\mbox{eigenvectors:~} [v_1 \cdots v_5]=
\begin{bmatrix}
\mathrm{e}^{\frac{2\pi j\times 1}{5}}& -1& 0& 0& 0\\
\mathrm{e}^{\frac{2\pi j\times 2}{5}}& 1& 0& 0& 0\\
\mathrm{e}^{\frac{2\pi j\times 3}{5}}& -2& -1& 1& 0\\
\mathrm{e}^{\frac{2\pi j\times 4}{5}}& -2&0&-1&0\\
\mathrm{e}^{\frac{2\pi j\times 5}{5}}&-2& 0& -2& 1
\end{bmatrix}.
\end{align*}
By (\ref{eq:dis_real}) we compute the control gain matrix
\begin{align*}
F =
\scalebox{0.7}{$\displaystyle
\begin{bmatrix}
-10.5 - 1.8164j& 2.5 - 1.8164j& 0&  0&  0\\
  0.3333 + 0.2422j&-3.4667 + 0.2422j& 0&  0&  0\\
 -0.4721 - 2.3511j&-0.4721 - 2.3511j&10& -2&  0\\
 -0.0442 - 0.4403j& 1.1679 - 0.4403j& 0&  0.697&  0\\
 -0.3979 + 0.4391j& 0.1427 + 0.4391j& 0& -0.5405& -0.0541
\end{bmatrix}.
$}
\end{align*}
Simulating the closed-loop system with initial condition $x(0) = [1+j \ \ 1-0.5j \ \ 1 \ \ j \ \ -1+j]^\top$, the result is displayed in Fig.~\ref{fig:pentagon}. Observe that a regular pentagon is formed, and the topology determined by $F$ contains a spanning tree.

\begin{figure}[t]
  	\centering
  	\includegraphics[width=0.45\textwidth]{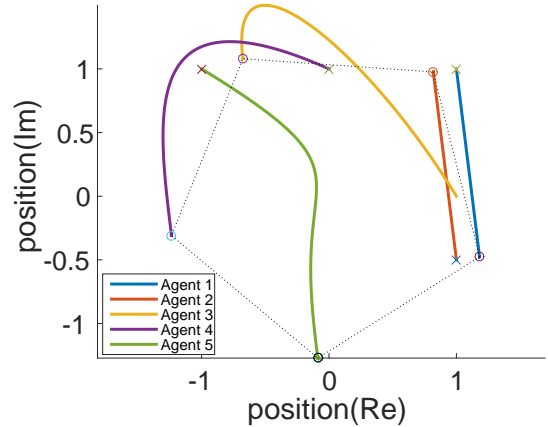}
	\caption{Scalable regular pentagon formation (x: initial positions, $\circ$: steady-state positions) \label{fig:pentagon}}
\end{figure}

\begin{figure}[t]
  	\centering
  	\includegraphics[width=0.45\textwidth]{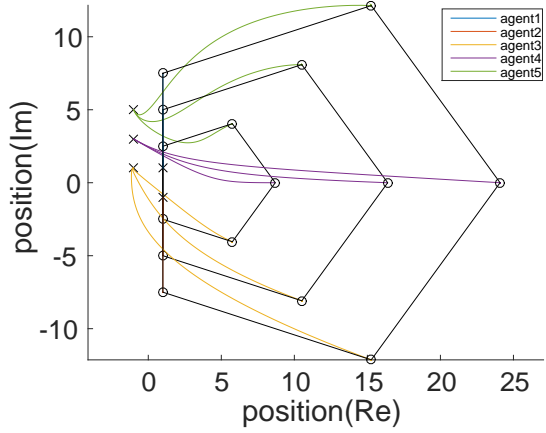}
	\caption{Rigid regular pentagon formation, with size $d=5,10,15$ (x: initial positions, $\circ$: steady-state positions) \label{fig:rigid}}
\end{figure}

Next, to achieve a rigid pentagon formation, we follow the method presented in Section~VI.A: first assign the following eigenstructure:
\begin{align*}
&\mbox{eigenvalues:~} \{\lambda_1,\ldots,\lambda_5\}=\{0,0,-1,-2,-3\}\\
&\mbox{eigenvectors:~} [v_1 \cdots v_5]=
\begin{bmatrix}
1&\mathrm{e}^{\frac{2\pi j\times 1}{5}}& 0& 0& 0\\
1&\mathrm{e}^{\frac{2\pi j\times 2}{5}}& 0& 0& 0\\
1&\mathrm{e}^{\frac{2\pi j\times 3}{5}}& 1& 0& 0\\
1&\mathrm{e}^{\frac{2\pi j\times 4}{5}}&0&1&0\\
1&\mathrm{e}^{\frac{2\pi j\times 5}{5}}& 0& 0& 1
\end{bmatrix}
\end{align*}
and by (\ref{eq:dis_real}) compute the control gain matrix
$$
F =
\scalebox{0.7}{$\displaystyle
\begin{bmatrix}
                 -8&                  0&   0&        0&        0\\
                  0 &  -3.1333 &   0&        0&        0\\
   0.618 + 1.9021j&  -2.618 - 1.9021j&   8&        0&        0\\
  -0.303 + 0.9326j&  -0.303 - 0.9326j&   0&   0.3939&        0\\
  -1.0614 + 0.7711j&   0.2506 - 0.7711j&   0&        0&  -0.3243
\end{bmatrix}.$}
$$
Thus the topology determined by $F$ is 2-rooted with nodes 1 and 2 the only two roots.
Then for different sizes ($d=5,10,15$), we obtain by (\ref{eq:rigidformation_control}) the control $u$.
Simulating the closed-loop system with the same initial condition $x(0)$ as above, the result is displayed in Fig.~\ref{fig:rigid}, where pentagons with specified sizes are formed.

Further, we consider the task of target encircling and solve it by circular motion introduced in Section VI.B.
Suppose that there is a static target, say $x_6$ with $\dot{x}_6=u_6$ ($u_6$ is constantly zero), and the goal is make the same 5 agents as above circle around $x_6$.  For this we treat the target $x_6$ as part of the multi-agent system; hence the augmented $A'$ and $B'$ are
\begin{align*}
A'&=\mbox{diag}(1.6,4.7, 3.0,-0.7,-4.2, 0) \\
B'&=\mbox{diag}(0.2,1.5,-0.5,-3.3,-3.7, 1). 
\end{align*}
Moreover, choose the following eigenstructure
\begin{align*}
&\scalebox{0.9}{$\displaystyle \mbox{eigenvalues:~} \{\lambda_1,\ldots,\lambda_6\}=\{0,j,-1,-2,-3,-4\}$}\\
&\mbox{eigenvectors:~} [v_1 \cdots v_6]=
\begin{bmatrix}
1&\mathrm{e}^{\frac{2\pi j\times 1}{5}}& 0& 0& 0& 0\\
1&\mathrm{e}^{\frac{2\pi j\times 2}{5}}& 1& 0& 0& 0\\
1&\mathrm{e}^{\frac{2\pi j\times 3}{5}}& 0& 1& 0& 0\\
1&\mathrm{e}^{\frac{2\pi j\times 4}{5}}& 0& 0& 1& 0\\
1&\mathrm{e}^{\frac{2\pi j\times 5}{5}}& 0& 0& 0& 1\\
1&0& 0& 0& 0& 0
\end{bmatrix}.
\end{align*}
Thus the desired formation is (again) a regular pentagon and the target $x_6$ is at the center of this pentagon. 
Moreover, by the eigenvectors the resulting topology will contain a spanning tree with $x_6$ (the target) being the root.
Corresponding to the formation vector is the eigenvalue $j$; hence the agents will perform circular motion at rate $1$. 
Since the center will not move ($x_6$ at the center is the root), this makes the first 5 agents encircle the target $x_6$.

By (\ref{eq:dis_real}) we compute the control gain matrix
$$
F =
\scalebox{0.6}{$\displaystyle 
\begin{bmatrix}
  -8 + 5j&    0&   0&        0&        0&          - 5j\\
  -0.428 + 0.84j& -3.8&   0&        0&        0&   1.0947 - 0.84j\\
   4.4116 - 0.7331j&    0&  10&        0&        0&  -8.4116 + 0.7331j\\
   0.5574 + 0.7795j&    0&   0&   0.697&        0&  -1.4664 - 0.7795j\\
  -0.5911 + 0.9447j&    0&   0&        0&  -0.0541&  -0.49 - 0.9447j\\
   0               &    0&   0&        0&        0&   0
\end{bmatrix}
$}
$$
to assign the above eigenstructure.
Indeed, the corresponding topology has a spanning tree whose root is $x_6$ and $\dot{x}_6=0$.
Simulating the closed-loop system with the initial condition $x(0) = [1-0.5j \ \  -2+2j \ \ -2+j \ \  -1+j \ \ -j \ \ 1+j]^\top$, the result is displayed in Fig.~\ref{fig:circular}. Observe that the target stays put at its initial position $1+j$, while other agents circle around it.

\begin{figure}[t]
  	\centering
  	\includegraphics[width=0.45\textwidth]{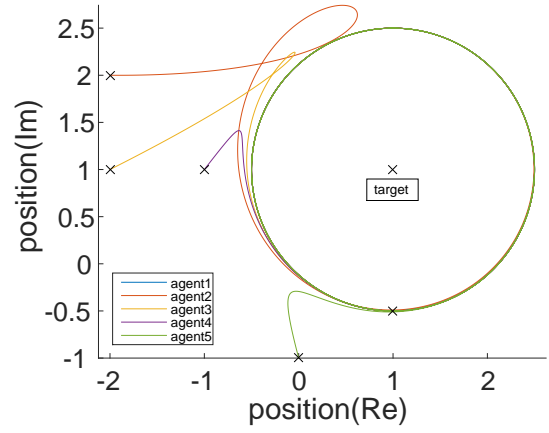}
	\caption{Target encircling by circular motion (x: initial positions, $\circ$: steady-state positions) \label{fig:circular}}
\end{figure}

Finally, we present an empirical study on the computation time of synthesizing feedback matrix $F$. In particular, we compare the centralized synthesis by (\ref{eq:dis_real}) and the hierarchical synthesis in Section~V; the result is listed in Table~\ref{tb:computation time} for different numbers of agents.\footnote{Computation is done by Matlab~R2014b on a laptop with Intel(R) Core(TM) i7-4510U CPU@2.00GHz 2.60GHz and 8.00GB memory.} Here for the hierarchical synthesis, we partition the agents in such a way that the number of groups and the number of agents in each group are `balanced' (to make $\tilde{n}$ small): e.g. 100 agents are partitioned into 10 groups of 10 agents each; 500 agents are partitioned into 16 groups of 23 agents each {\it plus} 6 groups of 22 each. 
Observe that the hierarchical synthesis is significantly more efficient than the centralized one, and the efficiency increases as the number of agents increases. In particular, for 1000 agents only 0.525 seconds needed, the hierarchical approach might well be sufficient for many practical purposes.

\begin{table}[t]
	\centering
	\caption{Comparison of computation time (unit: seconds)}
	\begin{tabular}{|c|c|c|} \hline
	agent \#& centralized method by (\ref{eq:dis_real}) & hierarchical method in Sec.~V\\
	\hline
	100&0.398&0.027 \\ \hline
	500&57.308&0.179 \\ \hline
	900&552.8419&0.394 \\ \hline
	1000&1068.729&0.525 \\ \hline
	\end{tabular}
	\label{tb:computation time}
\end{table}



\section{Concluding Remarks}

We have proposed a top-down, eigenstructure assignment based approach to synthesize state feedback control for solving multi-agent formation problems. The relation between the eigenstructures used in control synthesis and the resulting topologies among agents has been characterized, and special topologies have been designed by choosing appropriate eigenstructures. More general cases where the initial inter-agent coupling is arbitrary and/or there exist non-stabilizable agents have been studied, and a hierarchical synthesis procedure has been presented that improves computational efficiency. Further, the approach has been extended to achieve rigid formation and circular motion.  

In our view, the proposed top-down approach to multi-agent formation control is {\it complimentary} to the existing (mainstream) bottom-up approach (rather than {\it opposed} to). Indeed, the bottom-up approach, if successful, can produce scalable control strategies effective for possibly time-varying topologies, nonlinear agent dynamics, and robustness issues like communication failures, which are the cases very difficult to be dealt with by the top-down approach.  On the other hand, bottom-up design is generally challenging, requiring significant insight into the problem at hand and possibly many trial-and-errors in the design process; by contrast, top-down design is straightforward and can be automated by algorithms. Hence we suggest the following. When a control researcher or engineer faces a distributed control design problem for achieving some new cooperative tasks, one can start with a linear time-invariant version of the problem and try the top-down approach to derive a solution. With the ideas and insights gained from such a solution, one may then try the bottom-up design possibly for time-varying and nonlinear cases.

In future work, we aim to apply the top-down, eigenstructure assignment based approach to solve more complex cooperative control problems of multi-agent systems.  
In particular, our immediate goals are to achieve formations in three dimensions with obstacle avoidance abilities, as well as to deal with agents with higher-order (heterogeneous, possibly non-stabilizable) dynamics.


\section{Appendix}

We provide the proof of Proposition~\ref{prop:line topology}. 
For this we first briefly review the eigenstructure assignment method in \cite{Moore:77} for dealing with repeated eigenvalues and generalized eigenvectors.

\begin{lem}(\cite{Moore:77}) 
\label{lem:rep_eig}
Consider the system~(\ref{eq:standard dynamics}) and suppose that $(A,B)$ is controllable and Ker$B=0$. 
Let $\{d_1,\ldots,d_k\}$ ($k \leq n$) be a set of positive integers satisfying $\sum_{i=1}^k d_i =n$, 
$\{\lambda_1,\ldots,\lambda_k\}$ be a set of complex numbers, and $\{v_{11},...,v_{1d_1},\ldots,v_{k1},...,v_{kd_k}\}$ a set of linearly independent vectors in $\mathbb{C}^n$. Then there is a feedback matrix $F$ such that for every $i\in[1,k]$, 
\begin{align*}
&(\lambda_iI-A-BF)v_{i1}=0\\
&(\lambda_iI-A-BF)v_{ij}=v_{i(j-1)},~~~j=2,\ldots,d_i
\end{align*}
if and only if 
\begin{align} \label{eq:lemma2}
&(\exists w_{ij} \in \mathbb{C}^m) \notag\\ 
&\begin{bmatrix}
\lambda_i I - A \ \ -B
\end{bmatrix}
\begin{bmatrix}
  v_{i1}(\lambda_i) \\
  w_{i1}(\lambda_i)
\end{bmatrix} = 0 \notag \\
&\begin{bmatrix}
\lambda_i I - A \ \ -B
\end{bmatrix}\begin{bmatrix}
  v_{ij}(\lambda_i) \\
  w_{ij}(\lambda_i)
\end{bmatrix} = v_{i(j-1)},~~~j=2,\ldots,d_i.
\end{align}
\end{lem}

Lemma~\ref{lem:rep_eig} provides a necessary and sufficient condition for assigning repeated eigenvalues $\lambda_i$ ($i \in [1,k]$) with eigenvector $v_{i1}$ and generalized eigenvectors $v_{i2},...,v_{i d_i}$ (corresponding to a Jordan block of the closed-loop matrix $A+BF$).
When the condition holds and thus $F$ exists, $F$ may be constructed by the following procedure~\cite{Moore:77}.

\medskip

(i) Compute the following (maximal-rank) matrices
\begin{align*}
N_i=
\begin{bmatrix}
N_{1i}\\
N_{2i}
\end{bmatrix},\ \ \ 
S_i=
\begin{bmatrix}
S_{1i}\\
S_{2i}
\end{bmatrix}
\end{align*}
satisfying
\begin{align*}
[\lambda_iI-A~B]
\begin{bmatrix}
N_{1i}\\
N_{2i}
\end{bmatrix}=0,\ \ \ 
[\lambda_iI-A~B]
\begin{bmatrix}
S_{1i}\\
S_{2i}
\end{bmatrix}=-I.
\end{align*}
($S_i$ needs to be computed only when $d_i>1$.)

(ii) From the following vector chain
\begin{align*}
v_{i1}&=N_{1i}p_{i1}\\
v_{i2}&=N_{1i}p_{i2}-S_{1i}v_{i1}\\
&\ \vdots\\
v_{i d_i}&=N_{1i}p_{i d_i}-S_{1i}v_{i(d_i-1)}
\end{align*}
find the vectors $p_{i1},...,p_{id_i}$. Then
generate a new vector chain as follows:
\begin{align*}
w_{i1}&=-N_{2i}p_{i1}\\
w_{i2}&=-N_{2i}p_{i2}+S_{2i}v_{i1}\\
&\ \vdots\\
w_{id_i}&=-N_{2i}p_{id_i}+S_{2i}v_{i(d_i-1)}.
\end{align*}

(iii) Compute $F$ satisfying $Fv_{ij}=w_{ij}$ for all $i\in [1,k], j\in [1,d_i]$. (If no solution exists, alter one or more of the vectors $p_{ij}$ in Step~(ii).)

Now we are ready to prove Proposition~\ref{prop:line topology}.

{\it Proof of Proposition~\ref{prop:line topology}:} 
Consider the multi-agent system~(7) and assign the following eigenstructure:
\begin{align*}
&\mbox{eigenvalues: }  \lambda_1=0, \lambda_2 = \cdots = \lambda_n = -1 \notag\\
&\mbox{eigenvectors: }  [v_{11} \ v_{21} \cdots v_{2(n-1)}] =
\scalebox{0.7}{$\displaystyle
\left[
\begin{array}{cccc}
f_1	&0			&\cdots	&0 \\
f_2	&0			&\cdots	        &-f_2 \\
\vdots	 &\vdots	&	&\vdots \\
f_{n-1}	&0		&\reflectbox{$\ddots$}	        &-f_{n-1} \\
f_n	&-f_n			&\cdots	&-f_n
\end{array}
\right]
$}
\end{align*}

First, for $\lambda_1=0$ (with $d_1=1$), in Step~(i) we only need to compute $N_1$ and 
obtain $N_{11}=B$, $N_{21}=A$. Then in Step~(ii), from 
\begin{align*}
v_{11} = N_{11}p_{11}
\end{align*}
we have $p_{11}=B^{-1}v_{11}$. Hence 
\begin{align*}
w_{11} = -N_{21}p_{11} = -AB^{-1}v_{11} = [-\frac{a_1f_1}{b_1} \cdots -\frac{a_nf_n}{b_n}]^\top.
\end{align*}
It is verified that $w_{11}$, together with $v_{11}, \lambda_1$, satisfies the first equation of (\ref{eq:lemma2}).

Second, for $\lambda_2=-1$ (with $d_2=n-1$), in Step~(i) we derive $N_{12} = B$, $N_{22} = (A+I)$, $S_{12} = B$, 
$S_{22} = (A+I)-B^{-1}$. Then in Step~(ii), from the chain 
{\small \begin{align*}
v_{21}&=N_{12}p_{21}\\
v_{22}&=N_{12}p_{22}-S_{12}v_{21}\\
&\ \vdots\\
v_{2 (n-1)}&=N_{12}p_{2(n-1)}-S_{12}v_{2(n-2)}
\end{align*}
we find 
\begin{align*}
p_{21}&=[0 \cdots 0 \ \ -\frac{f_n}{b_n}]^\top\\
p_{22}&=[0 \cdots 0 \ \ -\frac{f_{n-1}}{b_{n-1}} \ \ -\frac{f_n}{b_n}-f_n]^\top\\
&\ \vdots\\
p_{2 (n-1)}&=[0 \ \ -\frac{f_{2}}{b_{2}} \ \ -\frac{f_3}{b_3}-f_3 \ \cdots \ -\frac{f_n}{b_n}-f_n]^\top.
\end{align*}
Hence we obtain
\begin{align*}
w_{21}&=N_{12}p_{21}\\
           &=[0 \cdots 0 \ \ \frac{(a_n+1)f_n}{b_n}]^\top\\
w_{22}&=N_{12}p_{22}-S_{12}v_{21}\\
           &=[0 \cdots 0 \ \ \frac{(a_{n-1}+1)f_{n-1}}{b_{n-1}} \ \ \frac{a_n f_n}{b_n}]^\top\\
&\ \vdots\\
w_{2 (n-1)}&=N_{12}p_{2(n-1)}-S_{12}v_{2(n-2)} \\
                 &\hspace{-0.9cm}=[0 \ \ \frac{(a_2+1)f_{2}}{b_{2}} \ \ \frac{a_3 f_3}{b_3} \ \cdots \ \frac{a_n f_n}{b_n}]^\top.
\end{align*}}
The chain $\{w_{21}, w_{22},..., w_{2(n-1)}\}$, together with $\{v_{21}, v_{22},..., v_{2(n-1)}\}$ and $\lambda_2 (=-1)$, is verified to satisfy (\ref{eq:lemma2}). Hence it follows from Lemma~\ref{lem:rep_eig} that the above eigenstructure can be assigned to the closed-loop matrix $A+BF$. Then by Proposition~\ref{prop:formation problem}, Problem~\ref{prob:formation problem} is solved. 

Finally we compute the feedback matrix $F$. 
Let $W = [w_{11} \ w_{21} \cdots w_{2(n-1)}]$. Since
\begin{align*}
V^{-1}=
\begin{bmatrix}
\frac{1}{f_1}	&0		&\cdots			&0			&0		\\
0		&0		&\cdots			&\frac{1}{f_{n-1}}	&-\frac{1}{f_n}	\\
		&		&\reflectbox{$\ddots$}	&\reflectbox{$\ddots$}	&		\\
0		&\frac{1}{f_2}		&\reflectbox{$\ddots$}			&			&		\\
\frac{1}{f_1}	&-\frac{1}{f_2}	&			&			&0		\\
\end{bmatrix}
\end{align*}
we obtain by $F=WV^{-1}$ that
\begin{align*}
F=
\begin{bmatrix}
\frac{-a_1}{b_1}	&0			&\cdots			&0			\\
\frac{f_2}{b_2f_1}	&\frac{-(1+a_2)}{b_2}	&			&			\\
			&\ddots			&\ddots			&			\\
0			&			&\frac{f_n}{b_nf_{n-1}}	&\frac{-(1+a_n)}{b_n}	
\end{bmatrix}.
\end{align*}
Therefore the closed-loop matrix is
\begin{align*}
A+BF =
\begin{bmatrix}
0 &0			&\cdots			&0			\\
\frac{f_2}{f_1}	&-1	&			&			\\
			&\ddots			&\ddots			&			\\
0			&			&\frac{f_n}{f_{n-1}}	&-1	
\end{bmatrix}
\end{align*}
and the corresponding graph $\mathcal{G}$ is a line topology. \hfill $\blacksquare$

\bibliographystyle{IEEEtran}
\bibliography{DistributedControl}

\end{document}